\let\origvec\vec
\documentclass[smallextended,nospthms]{svjour3}

\let\vec\origvec

\usepackage{pgfplots}
	\pgfplotsset{compat=1.12}

\usepackage{euscript}
\usepackage{marginnote}
\usepackage{anyfontsize} 
\usepackage{amsthm} 
\usepackage{mathrsfs, amsmath,  amssymb, euscript, amsfonts}
\usepackage{newtxtext,newtxmath} 
\usepackage{hyperref}
\hypersetup{
colorlinks= true
}
\usepackage[capitalize]{cleveref}
	\crefname{mtheorem}{Theorem}{}
	\crefname{equation}{}{}
	\crefformat{appendix}{\textsection#2#1#3}
	\crefformat{section}{\textsection#2#1#3}
	\crefrangeformat{equation}{(#3#1#4)-(#5#2#6)}

\newtheorem{theorem}{Theorem}
\newtheorem{mtheorem}{Theorem}

\newtheorem{lemma}[theorem]{Lemma}

\newtheorem{corollary}[theorem]{Corollary}

\theoremstyle{definition}
  \newtheorem{definition}[theorem]{Definition}
  \newtheorem{remark}[theorem]{Remark}
  
  \newtheorem{example}[theorem]{Example}

\AtBeginDocument{%
  \setlength{\oddsidemargin}{\dimexpr(\paperwidth-\textwidth)/2-1in}%
  \setlength{\evensidemargin}{\oddsidemargin}%
  \setlength{\topmargin}{%
    \dimexpr(\paperheight-\textheight)/2-\headheight-\headsep-1in}%
}

\usepackage{verbatim}
\newcommand{\R}{\mathbb{R}}
\newcommand{\C}{\mathbb{C}}
\newcommand{\N}{\mathbb{N}}
\newcommand{\Z}{\mathbb{Z}}
\newcommand{\T}{\mathbb{T}}

\newcommand{\eS}{\EuScript{S}}
\newcommand{\rH}{\mathcal{H}}

\newcommand{\rX}{\mathrm{X}}
\newcommand{\rL}{\mathrm{L}}
\newcommand{\rR}{\mathrm{R}}
\newcommand{\rK}{\mathrm{K}}
\newcommand{\sF}{\mathscr{F}}
\renewcommand{\Re}{\mathrm{Re}\,}

\newcommand{\textbi}[1]{\textit{\textbf{#1}}}

\newcommand{\tr}{\mathrm{tr}\,}
\newcommand{\ind}{\mathrm{ind}\,}
\newcommand{\sgn}{\mathrm{sgn}\,}
\newcommand{\ess}{\sigma_{\mathrm{ess}}}
\newcommand{\Arg}{\mathrm{Arg}\,}

\begin{document}

\title{The Witten Index for 1D Supersymmetric Quantum Walks with Anisotropic Coins}
\author{Akito Suzuki \and Yohei Tanaka}
\institute{Akito Suzuki \at
Division of Mathematics and Physics, 
Faculty of Engineering, Shinshu University, 
Wakasato, Nagano 380-8553, Japan
\email{akito@shinshu-u.ac.jp} \and 
Yohei Tanaka \at
School of Computer Science, Engineering and Mathematics, Flinders University, 
1284 South Road, Clovelly Park, 5042, SA, Australia \\
\email{tana0035@flinders.edu.au}
}
\maketitle

\begin{abstract}
Chirally symmetric discrete-time quantum walks possess supersymmetry, and their Witten indices can be naturally defined. The Witten index gives a lower bound for the number of topologically protected bound states. The purpose of this paper is to give a complete classification of the Witten index associated with a one-dimensional split-step quantum walk. It turns out that the Witten index of this model exhibits striking similarity to the one associated with a Dirac particle in supersymmetric quantum mechanics.
\end{abstract}



\section{Introduction}
\label{Section: Introduction}


Discrete-time quantum walks are versatile platforms realising topological phenomena \cite{Gross,AsOb13,OA,OK,MKO,XZB,Ced1,Ced2}. Kitagawa et al. \cite{KRBD} proposed a split-step quantum walk with chiral symmetry and experimentally observed topologically protected bound states \cite{KBFRBKAD}
(see \cite{Kit} for a comprehensive review). 
For such bound states, Fuda et al. \cite{FFSd,FFS1} proved the robustness against compact perturbations and the spatial exponential decay property with mathematical rigour. Barkhofen et al. \cite{BLNSS} implemented a chiral symmetric discrete-time quantum walk with supersymmetry. Recently,  the first author of the present paper \cite{Suzuki18} proved that all the chiral symmetric quantum walks possess supersymmetry and that a discrete-time quantum walk has chiral symmetry if and only if the product of two unitary involutions represents its evolution operator. From these facts, we know that a discrete-time quantum walk can possess supersymmetry even if it does not have apparent chiral symmetry. Indeed, all homogeneous one-dimensional two-state quantum walks \cite{ABNVW01,Ko02}, multi-dimensional quantum walks \cite{FFSd}, various types of quantum walks on graphs \cite{MNRoS07,MNRS09,Se13,HKSS14,Po16,HiSe1,HiSe2,HSS,KPSS18}, and several quantum-walk based algorithms \cite{Gr,Sz} have evolution operators that can be represented by the product of two unitary involutions, and therefore they exhibit supersymmetry. 
See \cite{Suzuki18} for more details, and \cite{Oh} for many examples of inhomogeneous one-dimensional quantum walks \cite{Ko09,Ko10,ShiKa10,KoLuSe13,EEKST1} whose evolutions are written by the product of two unitary evolutions. 

As shown in \cite{Suzuki18}, 
a supersymmetric quantum walk (SUSYQW) assigns the Witten index,
which provides a lower bound for the number of topologically protected bound sates. In this paper, we classify the Witten index for the split-step quantum walk entirely. 

\subsection{Witten index for SUSYQWs}
\label{Subsection:Witten}
To give a precise definition of the Witten index for 
SUSYQWs 
introduced in \cite{Suzuki18}, we briefly review here the supersymmetric structure of chiral symmetric quantum walks. 
We say that a unitary operator $U$ on a Hilbert space $\rH$ has chiral symmetry if there exists a unitary involution $\varGamma$ on $\rH$
(i.e., $\varGamma^{-1} = \varGamma^* = \varGamma$) such that
$\varGamma U \varGamma = U^{-1}$. 
$U$ has chiral symmetry if and only if 
it can be represented as a product of two unitary involutions 
$\varGamma$ and $C:=\varGamma U$. 

Suppose that $U$ is the evolution of a chiral symmetric quantum walk. Namely, there are two unitary involutions $\varGamma$ and $C$ such that $U=\varGamma C$. We call
$Q:= [\varGamma, C]/2i$
a supercharge and $H=Q^2$ the superhamiltonian, 
where $[X, Y] :=XY-YX$ is the commutator. 
A direct calculation proves that $\varGamma$ and $H$ commute and
hence $H$ can be decomposed into $H_+\oplus H_-$
with respect to the decomposition 
$\ker (\varGamma-1) \oplus \ker(\varGamma+1)$. 
We now define a topological index $ {\rm ind}(\varGamma, C)$ so that it coincides with the Witten index 
$\Delta_\varGamma(H)$ of $H$ with respective to $\varGamma$, i.e.,
\begin{equation}
\label{eq:witt} 
{\rm ind}(\varGamma, C) =
\Delta_\varGamma(H) := \dim \ker H_+ - \dim \ker H_-. 
\end{equation} 

In this sense, we call a pair $(\varGamma, C)$ of two unitary involutions 
a SUSYQW 
with the evolution $U=\varGamma C$
and call the index ${\rm ind}(\varGamma, C)$ 
the Witten index for the SUSYQW. 
We say that a SUSYQW 
$(\varGamma, C)$ 
is Fredholm if $H$ is Fredholm. 
As shown in \cite{Suzuki18}, 
the Fredholmness of $(\varGamma, C)$ depends only on $U$
(or equivalently $H$), and it is independent of the choice of $(\varGamma, C)$. However, the Witten index ${\rm ind}(\varGamma, C)$ 
depends on the choice of $(\varGamma, C)$.
If $(\varGamma, C)$ is Fredholm, 
then the index ${\rm ind}(\varGamma, C)$ is robust against compact perturbations (see \cref{Section: Appendix}). 

\subsection{Main result}

In this paper, we study
a split-step quantum walk \cite{FFS1,FFSd,FFS1w}, which unifies Kitagawa's split-step quantum walk \cite{Kit} and a usual one-dimensional quantum walk \cite{ABNVW01,Ko02,Suzuki16}. 
Let  $ \rH: = \ell^2(\Z,\C^2)$ be the state space of the split-step quantum walk. Idetifying $\rH$ with $\ell^2(\mathbb{Z}) \oplus \ell^2(\mathbb{Z})$, we define a shift operator $\varGamma$ on $\rH$ as
\begin{equation}
\label{shift} 
\varGamma 
= \begin{pmatrix}
    p & q L \\ q^* L^* & - p \end{pmatrix},
\end{equation}
where $L$ is the left shift operator on $\ell^2(\mathbb{Z})$. We suppose that $(p,q) \in \mathbb{R} \times \mathbb{C}$ satisfies $p^2 +|q|^2=1$, which ensures that $\varGamma$ is a unitary involution. 
We define a coin operator on $\rH$ as
\begin{equation}
\label{coin} 
C = \begin{pmatrix} a_1 & b^* \\ b & a_2 \end{pmatrix}, 
\end{equation}
where $a_1$, $a_2$ and $b$ are the multiplication operator on $\ell^2(\mathbb{Z})$ by functions $a_j:\mathbb{Z} \to \mathbb{R}$ ($j=1,2$) and $b:\mathbb{Z} \to \mathbb{C}$. We assume $a_j(x)^2 + |b(x)|^2 = 1$ ($x \in \mathbb{Z}$, $j=1,2$), which garantees that $C$ is a unitary involution.
The evolution operator of the split-step quantum walk is defined as the product of $\varGamma$ and $C$, i.e., $U = \varGamma C$. 
Since $\varGamma$ and $C$ are unitary involutions,
the evolution $U$ has chiral symmetry, and
$(\varGamma,C)$ defines a 
SUSYQW 
as explained in Subsection \ref{Subsection:Witten}. 

As shown in \cite{Suzuki18}, the modulus of ${\rm ind}(\varGamma, C)$ provides the lower bound for the number of topological bound states. Therefore, if 
 ${\rm ind}(\varGamma, C)$ is nonzero, 
the corresponding quantum walk with the evolution $U = \varGamma C$  has a topological bound state. Motivated by this fact, we give a complete classification of the Witten index for the split-step quantum walk
 $(\varGamma, C)$ 
 defined by \eqref{shift} and \eqref{coin}. 

To state our main result, we suppose that the coin 
$C$ is anisotropic \cite{Richard-Suzuki-Aldecoa18,Richard-Suzuki-Aldecoa19}, i.e., $C(x) :=  \begin{pmatrix} a_1(x) & b^*(x) \\ b(x) & a_2(x) \end{pmatrix}$ has limits as $x \to \pm \infty$. We denote these limits by
$C(\rL) = \lim_{x \to -\infty} C(x)$ and $C(\rR) = \lim_{x \to +\infty} C(x)$. 
Clearly, the limit coins are unitary and hermitian. 
We say that a unitary and hermitian matrix is trivial if it equals $+1$ or $-1$. 
If the limit coins $C(\sharp)$ ($\sharp =\rL, \rR$) are nontrivial, they can be assumed to be a unitary involution of the form
 \[ C(\sharp) =  \begin{pmatrix} a(\sharp) & b^*(\sharp) \\ b(\sharp) & -a(\sharp) \end{pmatrix} \]
without loss of generality (see Section \ref{Section: Main Results} for more details). We are now in a position to state our main result.

\begin{mtheorem}
\label{Theorem: MainTheorem}
Let $\varGamma$ and $C$ be 
defined by \eqref{shift} and \eqref{coin}.
Suppose that $C$ is anisotropic and the limit coins $C(\sharp)$ 
($\sharp = \rL,\rR$) are nontrivial.  
Then 
\begin{equation}
\label{Equation: Fredholmness}
\tag{A1}
\mbox{$(\varGamma,C)$ is Fredholm if and only if $|p| \neq |a(\sharp)|$ for each $\sharp = \rL,\rR.$}
\end{equation}
In this case, we have
\begin{equation}
\label{Equation: Index Formula for Type III}
\tag{A2}
\ind(\varGamma,C) = 
\begin{cases}
+ \sgn p, & |a(\rR)| < |p| <|a(\rL)|, \\
- \sgn p, & |a(\rL)| < |p| < |a(\rR)|, \\
0, & \mbox{otherwise}. 
\end{cases}
\end{equation}
\end{mtheorem}
The case where at least one of the two limits $C(\rL)$ and $C(\rR)$ is a trivial unitary involution is excluded 
here, since the pair $(\varGamma,C)$ with this property automatically fails to be Fredholm (\cref{Lemma: Infinite Dimensional Ker Qepsilon}). \cref{Theorem: MainTheorem} provides a necessary and sufficient condition for the Fredholmness of 1D split-step SUSYQWs endowed with anisotropic coins, together with complete classification of the associated Witten index.

Note also that the model takes its simplest form when $p=0,$ but the associated Witten index is $0$ in this case by the formula \cref{Equation: Index Formula for Type III}. It is therefore important to consider non-zero $p$ as well as the trivial case $p=0.$  

There are close links between quantum walks and Dirac particles. In a continuous limit, quantum walks converge to Dirac particles \cite{BES,Str06} (see \cite{MS} for a mathematically rigorous and general proof). 
Klein's paradox and Zitterbewegung in quantum walks were found in \cite{Mey97,Str07,Ku08}.   
\cref{Theorem: MainTheorem} inspires a new relation between quantum walks and Dirac particles in comparison with the rusult of Bole et al. \cite{BGGSS87}: for the Dirac operator $Q = -i \sigma_2 d/dx + \sigma_1 \phi(x)$ on $L^2(\mathbb{R}) \oplus L^2(\mathbb{R})$ with an anisotropic scalar potential $\phi(x)$ satisfying $\lim_{x \to \pm \infty}\phi(x) = \phi_\pm \in \mathbb{R}$, the Witten index equals $\pm 1$ if $\pm \phi_- < 0 < \pm \phi_+$ and it equals  0 otherwise. 

\subsection{Organisation and strategy of the paper}

The present paper is organised as follows. In \cref{Section: Preliminaries} we go through some preliminary results including the precise definition of the one-dimensional split-step SUSYQW $(\varGamma,C)$. It is shown in \cref{Section: Diagonalisation} that the Witten index of $(\varGamma,C)$ is given by the Fredholm index of a certain well-defined operator $Q_{\epsilon_+}$ on $\ell^2(\Z)$ (see \cref{Theorem: Witten Index Formula} for details);
\begin{equation}
\label{Equation: Naive Definition of Witten Index}
\ind(\varGamma,C) = \ind Q_{\epsilon_+} = \dim \ker Q_{\epsilon_+} - \dim \ker Q_{\epsilon_+}^*.
\end{equation}
We show that the operator $Q_{\epsilon_+}$ is of the form $Q_{\epsilon_+} = \alpha L + \alpha'L^* + \beta,$ where $\alpha,\alpha',\beta$ are $\C$-valued sequences indexed by $\Z$ and $L,L^*$ are the left and right shift operators on $\ell^2(\Z)$ respectively. In \cref{Section: Classification of Dimensions} we separately compute the two dimensions on the right hand side \cref{Equation: Naive Definition of Witten Index}. With the explicit form of $Q_{\epsilon_+}$ mentioned above in mind, we shall end up solving second-order linear difference equations of the form 
\begin{equation}
\label{Equation: Naive DE}
\alpha(x) \Psi(x+1) + \alpha'(x)\Psi(x-1) + \beta(x) \Psi(x) = 0,
\qquad \Psi = (\Psi(x))_{x \in \Z} \in \ell^2(\Z),
\end{equation}
which is known to have two linearly independent algebraic solutions. Here, we need not only to algebraically solve Equation \cref{Equation: Naive DE}, but also to ensure the solutions $\Psi$ to be square summable. This is precisely why the difference on the right-hand side of \cref{Equation: Naive Definition of Witten Index} can still be non-zero.

In \cref{Section: Proof of the main theorem} we prove \cref{Theorem: MainTheorem} by making use of the index formula \cref{Equation: Naive Definition of Witten Index}. The present paper concludes with \cref{Section: Concluding Remarks}, the main focus of which is a possible generalisation of the Witten index associated with SUSYQWs which fail to be Fredholm.  Finally, \cref{Section: Appendix} contains a brief summary of the several invariance principles of the Witten index, each of which plays a supplementary role in this paper.

\section{Preliminaries}
\label{Section: Preliminaries}
The primary focus on the present paper is discrete-time quantum walks, and so we shall henceforth assume that all (linear) operators in this paper are everywhere-defined bounded operators.

\subsection{A brief overview of supersymmetry}
\label{Section: Supersymmetry}
Here, we give a brief overview of supersymmetry by going through some preliminary results in a somewhat rapid manner. What follows can be found in any standard textbook on the subject (see, for example, \cite[\textsection 5]{Book:Bernd92:TheDiracEquation} or \cite[\textsection 7.13]{Book:Arai17:AnalysisOnFockSpacesAndMathematicalTheoryOfQuantumFields}), and so proofs are omitted. An abstract operator $\varGamma$ on a Hilbert space $\rH$ is called an \textbi{involution}, if $\varGamma^2 = 1.$ Note that if an operator possesses any two of the properties ``involutory'', ``unitary'' and ``self-adjoint'', then it possesses the third. We shall make use of the following finite-dimensional example throughout this paper;

{\footnotesize 
\begin{example}[$2 \times 2$ case]
\label{Example: Unitary Involutory Matrix of Dimension 2}
A $2 \times 2$ matrix $C$ is a unitary involution if and only if it is of the following form:
\begin{equation}
\label{Equation: Unitary Involutory Matrix of Dimension 2}
C = 
\begin{pmatrix}
a_1 & b^* \\
b & a_2
\end{pmatrix},
\end{equation}
where the triple $(a_1,a_2,b) \in \R \times \R \times \C$ satisfies 
\[
b(a_1 + a_2) = 0 \mbox{ and } a_j^2 + |b|^2 = 1, \qquad j=1,2.
\]
In particular, $C = -1$ or $C = +1,$ which will be referred to as \textbi{trivial unitary involutions}, satisfies all of the above equalities. It is then easy to observe that a $2 \times 2$ matrix $C$ is a non-trivial unitary involution if and only if it is of the following form:
\begin{equation}
C = 
\begin{pmatrix}
a & b^* \\
b & -a
\end{pmatrix} \mbox{ and } a^2 + |b|^2 = 1.
\end{equation}
\end{example}
}

A self-adjoint operator $Q$ on $\rH$ is called a \textbi{supercharge} with respect to a unitary involution $\varGamma,$ if it satisfies the anti-commutation relation $Q\varGamma+\varGamma Q = 0,$ where the left hand side is commonly denoted by the symbol $\{Q,\varGamma\}.$ With a canonical decomposition $\rH = \rH_+ \oplus \rH_-$ by $\rH_{\pm} := \ker(\varGamma \mp 1)$ in mind, a supercharge $Q$ and the \textbi{superhamiltonian} $H := Q^2$ admit the following block-operator representations respectively;
\begin{align}
\label{Equation: Matrix Representation of Q}
Q &= 
\begin{bmatrix}
0  & Q_-  \\
Q_+  &  0
\end{bmatrix}, \mbox{ where $Q_\pm : \rH_\pm \to \rH_\mp$ satisfy $Q_\pm^* = Q_\mp,$}\\
H 
&=
\begin{bmatrix}
H_+  & 0  \\
0  &  H_-
\end{bmatrix}, \mbox{ where $H_\pm : \rH_\pm \to \rH_\pm$ satisfy $H_\pm^* = H_\pm.$}
\end{align}
The superhamiltonian $H$ simultaneously represents two non-negative hamiltonians $H_+,H_-$ whose spectra are identical except possibly for $0.$ The \textbi{Witten index} of the superhamiltonian $H$ with respect to $\varGamma$ is given by
\begin{equation}
\label{Equation: Definition of Witten Index}
\Delta_\varGamma(H) := \dim \ker H_+ - \dim \ker H_-,
\end{equation}
which measure the difference in the number of zero-energy ground states of $H_+,H_-$, whenever the right-hand side is well-defined. Recall that $Q_+$ is a Fredholm operator if and only if both $\dim \ker Q_\pm$ are finite-dimensional and the range of $Q_+$ is closed. In this case, the \textbi{Fredholm index} of $Q_+$ is defined by
$
\ind Q_+ := \dim \ker Q_+ - \dim \ker Q_-.
$ 
The following two results about a general bounded operator $A : \rH \to \rK$ are useful: 
\begin{align}
\label{Equation: Kernel Theorem} &\ker A^*A = \ker A, \\
\label{Equation: Closed Range Theorem}
&\mbox{$\inf \sigma(A^*A) \setminus \{0\} > 0$ if and only if $A$ has a closed range,}
\end{align}
where the proof of \cref{Equation: Closed Range Theorem} can be found, for example, in \cite[Lemma 7.27]{Book:Arai17:AnalysisOnFockSpacesAndMathematicalTheoryOfQuantumFields}). With \cref{Equation: Kernel Theorem} in mind, the Witten index has a precise interpretation as $\Delta_\varGamma(H) = \ind Q_+,$ provided that $Q_+$ is a Fredholm operator.

Following \cite{Suzuki18} we introduce the supercharge associated with 
a supersymmetric quantum walk (SUSYQW);
\begin{definition}
We call a pair $(\varGamma,C)$ of two unitary involutions 
on a Hilbert space $\rH$ a SUSYQW with the evolution operator
$U=\varGamma C$. 
\end{definition}
For a SUSYQW $(\varGamma,C)$, 
$Q := [\varGamma,C]/2i$ is a supercharge with respect to 
$\varGamma$, i.e., $\{Q, \varGamma\}=0$. 
We define the Witten inex of the SUSYQW as 
\[ {\rm ind}(\varGamma,C) = \Delta_\varGamma(H), \]
where the right-hand side is defined by \cref{Equation: Definition of Witten Index} with the superhamiltonian $H=Q^2$ for the supercharge $Q$
of the SUSYQW. 
\begin{definition}[Fredholmness]
A SUSYQW $(\varGamma,C)$ is said to be \textbi{Fredholm}, if $Q_+$ as in \cref{Equation: Matrix Representation of Q} is a Fredholm operator.
\end{definition}
\subsection{Definition of the model}
\label{Section: Main Results}

Given $\rX = \C$ or $\rX = \C^2,$ we shall consider the Hilbert space of square-summable $\rX$-valued sequences:
\[
\ell^2(\Z,\rX) := 
\{\Psi : \Z \to \rX \mid \sum_{x \in \Z} \|\Psi(x)\|^2_\rX < \infty \},
\]
where $\|\cdot\|_\rX$ is the standard norm defined on $\rX.$ We shall agree to write elements of $\C^2$ as $2 \times 1$ column vectors. With this convention in mind, an element $\Psi$ of $\ell^2(\Z,\C^2)$ is written by $\Psi = (\Psi_1,\Psi_2)^\mathrm{T},$ where $\Psi_1,\Psi_2 \in \ell^2(\Z,\C^2).$ On the Hilbert space $\ell^2(\Z) := \ell^2(\Z,\C),$ the left-shift operator $L$ and the right-shift operator $L^*$ are given respectively by
\[
(L\Psi)(x) = \Psi(x + 1) \mbox{ and } (L^*\Psi)(x) = \Psi(x - 1), \qquad x \in \Z.
\]
Evidently, we have $LL^* = L^*L = 1.$ Let $\rH = \ell^2(\Z,\C^2)$ be the state space of a quantum walker throughout the present paper. With the canonical identification $\rH = \ell^2(\Z) \oplus \ell^2(\Z)$ in mind, we are now in a position to introduce the precise definition of the model we shall consider throughout this paper;
\begin{definition}
\label{Definition: Split-step SUSYQW}
A \textbi{(one-dimensional) split-step SUSYQW} is a pair $(\varGamma,C)$ of two unitary involutions on $\rH$ that are of the following forms;
\begin{equation}
\label{Equation2: Definition of Split-Step QW}
\varGamma = 
\begin{pmatrix}
p & qL \\
q^* L^* & -p
\end{pmatrix}_{\ell^2(\Z) \oplus \ell^2(\Z)} \mbox{ and }
C = 
\begin{pmatrix}
a_1 & b^* \\
b & a_2
\end{pmatrix}_{\ell^2(\Z) \oplus \ell^2(\Z)},
\end{equation}
where the pair $(p,q) \in \R \times (\C \setminus \{0\})$ and the triple $(a_1,a_2,b)$ of $\C$-valued sequences satisfy all of the following conditions:
\begin{align}
\label{Equation2: Condition on Shift} &\theta = \Arg q, \\
\label{Equation: Condition on C} &p^2 + |q|^2 = 1, \\
\label{Equation2: Condition on C} &a_j(x)^2 + |b(x)|^2 = 1, \qquad j=1,2, \\
\label{Equation1: Condition on C} &b(x)(a_1(x) + a_2(x)) = 0, 
\end{align}
where we assume that both \cref{Equation2: Condition on C} and \cref{Equation1: Condition on C} hold true for each $x \in \Z.$ Note that the sequences $a_j = (a_j(x))_{x \in \Z}$ and $b = (b(x))_{x \in \Z}$ here are canonically identified with their associated multiplication operators on $\ell^2(\Z).$
\end{definition}


\begin{definition}[anisotropic coins]
\label{Definition: Anisotoropic Coin}
Let $\rL = - \infty,$ and let $\rR = + \infty.$ Let $(\varGamma,C)$ be a split-step SUSYQW. The coin operator $C$ is called an \textbi{anisotropic coin}, if it admits the following two-sided limits:
\begin{equation}
\lim_{x \to \sharp} C(x) = C(\sharp) =
\begin{pmatrix} 
a_1(\sharp) & b(\sharp)^* \\
b(\sharp) & a_2(\sharp) 
\end{pmatrix}_{\ell^2(\Z) \oplus \ell^2(\Z)},
\qquad \sharp = \rL,\rR,
\end{equation}
where we assume that \cref{Equation1: Condition on C} and \cref{Equation2: Condition on C} both hold true for each $x = \rL,\rR.$ Note that if $C(\sharp)$ is a non-trivial unitary involution, then we shall assume without loss of generality (see \cref{Example: Unitary Involutory Matrix of Dimension 2} for details) that
\begin{equation}
\label{Equation: Property of Nontrivial Unitary Involutions}
a(\sharp) := a_1(\sharp) = -a_2(\sharp).
\end{equation}
\end{definition}

As in \cref{Definition: Anisotoropic Coin} we shall always let $\rL = - \infty$ and $\rR = + \infty$ throughout this paper. This commonly used convention is, for example, in accordance with \cite{Richard-Suzuki-Aldecoa18}.

\section{Diagonalisation} 
\label{Section: Diagonalisation}

\subsection{The main result}
The ultimate purpose of the current section is to prove the following index formula for the Witten index;
 
\begin{theorem}
\label{Theorem: Witten Index Formula}
Let $(\varGamma,C)$ be a split-step SUSYQW, where $C$ may or may not be anisotropic. Then there exists a unitary operator $\epsilon$ on $\rH$ such that the supercharge $2iQ := [\varGamma,C]$ admits off-diagonalisation of the following form with respect to the orthogonal decomposition $\rH = \ell^2(\Z) \oplus \ell^2(\Z):$
\begin{equation}
\label{Equation: Off-diagonalisation of Qepsilon}
\epsilon^* Q \epsilon
=
\begin{pmatrix}
0 & Q_{\epsilon_-} \\
Q_{\epsilon_+} & 0
\end{pmatrix}_{\ell^2(\Z) \oplus \ell^2(\Z)},
\end{equation}
where the three operators $\epsilon, Q_{\epsilon_+}, Q_{\epsilon_-}$ are given respectively by
\begin{align}
\epsilon &= 
\label{Equation: Definition of Epsilon}
\frac{1}{\sqrt{2}} 
\begin{pmatrix}
\sqrt{1 + p}  & -\sqrt{1 - p}  \\
\sqrt{1 - p} e^{-i \theta} L^* & \sqrt{1 + p}e^{-i \theta} L^*
\end{pmatrix}_{\ell^2(\Z) \oplus \ell^2(\Z)}, \\
\label{Equation: Explicit Definition of Qepsilon}
-2iQ_{\epsilon_\pm} &= (1 \pm p) e^{i \theta}Lb - (1 \mp p)e^{-i \theta}b^*L^* \pm |q|(a_2(\cdot + 1) - a_1).
\end{align}
Furthermore, the split-step quantum walk $(\varGamma,C)$ is Fredholm if and only if $Q_{\epsilon_+}$ is a Fredholm operator. In this case, we have
\begin{equation}
\label{Equation: Wintten Index Formula for Fredhoml Pairs}
\ind(\varGamma,C) = \ind Q_{\epsilon_+} = \dim \ker Q_{\epsilon_+} - \dim \ker Q_{\epsilon_-}.
\end{equation}
\end{theorem}

\begin{remark}
A direct computation shows that the supercharge $Q$ itself is not representable as an off-diagonal matrix with respect to the $\ell^2(\Z)$-decomposition $\rH = \ell^2(\Z) \oplus \ell^2(\Z)$, unlike the standard representation \cref{Equation: Matrix Representation of Q} which makes use of the canonical decomposition $\rH = \rH_+ \oplus \rH_-.$ To avoid confusion, we shall henceforth adhere to the convention that the round parentheses are used in the former representations, whereas the square parentheses are used in the latter representations.
\end{remark}

\subsection{The significance of diagonalisation}

The main result of the current section, \cref{Theorem: Witten Index Formula}, might look rather technical at first glance, but as we shall see shortly the basic idea behind the proof is nothing but simple diagonalisation of the shift operator as in the following lemma;

\begin{lemma}
\label{Theorem: Diagonalisation of the Shift Operator}
Let $(\varGamma,C)$ be a split-step SUSYQW. The operator $\epsilon$ given by \cref{Equation: Definition of Epsilon} is a unitary operator which diagonalises the shift operator $\varGamma$ as follows:
\begin{equation}
\label{Equation: Diagonalisation of the Shift Operator}
\epsilon^* \varGamma \epsilon =
\begin{pmatrix}
1 & 0 \\
0 & -1
\end{pmatrix}.
\end{equation}
\end{lemma}
\begin{proof}
It is left as an easy exercise for the reader to verify that $\epsilon$ is unitary, and that the following two equalities hold true:
\begin{align}
\label{Equation1: Unitary Transform by Epsilon}
2\epsilon^*
\begin{pmatrix}
X & 0 \\
0 & X'
\end{pmatrix}
\epsilon &=
\begin{pmatrix}
(1 + p)X + (1-p)LX'L^* & - |q|(X - LX'L^*) \\
- |q|(X - LX'L^*) & (1-p)X + (1 + p)LX'L^*
\end{pmatrix}, \\ 
\label{Equation2: Unitary Transform by Epsilon}
2\epsilon^*
\begin{pmatrix}
0 & Y' \\
Y & 0
\end{pmatrix}
\epsilon &=
\begin{pmatrix}
qLY + q^*Y'L^* & \frac{-(1-p) qLY + (1 + p) q^*Y'L^*}{|q|} \\
\frac{(1 + p)qLY - (1 - p) q^*Y'L^*}{|q|} & -qLY - q^*Y'L^*
\end{pmatrix}.
\end{align}
With these two equalities in mind, we obtain \cref{Equation: Diagonalisation of the Shift Operator} as follows:
\begin{align*}
2\epsilon^* \varGamma \epsilon
&=
2\epsilon^*
\begin{pmatrix}
p & 0 \\
0 & -p
\end{pmatrix}
\epsilon
+
2\epsilon^*
\begin{pmatrix}
0 & qL \\
q^*L^* & 0
\end{pmatrix}
\epsilon \\
&=
\begin{pmatrix}
2p^2 & -2p|q| \\
-2p|q| & -2p^2
\end{pmatrix}
+
\begin{pmatrix}
2|q|^2 & 2p|q| \\
2p|q| & -2|q|^2
\end{pmatrix} =
\begin{pmatrix}
2 & 0 \\
0 & -2
\end{pmatrix}.
\end{align*}
\end{proof}

\begin{remark}
The diagonalisation of the form \cref{Equation: Diagonalisation of the Shift Operator} is not unique. Indeed, as the experienced reader might immediately notice, one can introduce the discrete Fourier transform $\sF$ following \cite{GJS} and consider following unitary transform;
\begin{equation}
\label{Equation: Fourier Transform of Shift} 
\sF  \varGamma \sF^{-1} 
=
\begin{pmatrix}
p & q e^{i(\cdot)}\\
q^* e^{-i(\cdot)} & -p
\end{pmatrix},
\end{equation}
where the right-hand side is diagonalisable in infinitely many different ways. Since the transform \cref{Equation: Fourier Transform of Shift} is reversible, we can then obtain diagonalisation of the form \cref{Equation: Diagonalisation of the Shift Operator}. The unitary operator $\epsilon$ given explicitly by \cref{Equation: Definition of Epsilon} is constructed in this precise manner.
\end{remark}

In what follows, we shall make use of the unitary invariance of the Witten index as in \cref{Theorem: Invariance of the Witten Index}. Let us fix an arbitrary unitary operator $\epsilon$ which gives the diagonalisation \cref{Equation: Diagonalisation of the Shift Operator}. We can then consider a new unitarily equivalent SUSYQW given by $(\varGamma_\epsilon,C_\epsilon) := (\epsilon^*\varGamma\epsilon, \epsilon^* C \epsilon).$ Since the new shift operator $\varGamma_\epsilon$ is given by \cref{Equation: Diagonalisation of the Shift Operator}, we see immediately that the two subspaces $\rH_{\epsilon_\pm} := \ker(\varGamma_\epsilon \mp 1)$ are given respectively by 
\[
\rH_{\epsilon_+} = \ell^2(\Z) \oplus \{0\} \mbox{ and } \rH_{\epsilon_-} = \{0\} \oplus \ell^2(\Z).
\]
Since the two subspaces $\rH_{\epsilon_\pm}$ can be canonically identified with $\ell^2(\Z),$ the following abstract version of \cref{Theorem: Witten Index Formula} holds true;

\begin{lemma}
\label{Lemma: Diagonalisation of Shift Implies Off-diagonalisation of Super Charge}
Let $(\varGamma,C)$ be a split-step SUSYQW, and let $\epsilon$ be any unitary operator which gives diagonalisation \cref{Equation: Diagonalisation of the Shift Operator}. Then the new supercharge $Q_\epsilon := \epsilon^* Q \epsilon$ admits the following off-diagonal block matrix representation with respect to the decomposition $\rH = \ell^2(\Z) \oplus \ell^2(\Z):$
\begin{equation}
\label{Equation: Skew-diagonal Representation of Qepsilon}
Q_\epsilon = 
\begin{pmatrix}
0 & Q_{\epsilon_-} \\
Q_{\epsilon_+} & 0
\end{pmatrix}_{\ell^2(\Z) \oplus \ell^2(\Z)}.
\end{equation}
Furthermore, $(\varGamma,C)$ is a Fredholm operator if and only if  $Q_{\epsilon_+} : \ell^2(\Z) \to \ell^2(\Z)$ is a Fredholm operator. In this case, we have
\[
\ind(\varGamma,C) = \ind Q_{\epsilon_+} = \dim \ker Q_{\epsilon_+} - \dim \ker Q_{\epsilon_-}.
\]
\end{lemma}
\begin{proof}
As in \cref{Section: Supersymmetry} the new supercharge $Q_\epsilon = \epsilon^* Q \epsilon$ admits
\[
Q_\epsilon = 
\begin{bmatrix}
0 & Q'_{\epsilon_-} \\
Q'_{\epsilon_+} & 0
\end{bmatrix}_{\rH_{\epsilon_+} \oplus \rH_{\epsilon_-}}.
\]
Observe first that $\ell^2(\Z)$ can be canonically identified with $\rH_\pm$ by the unitary operators $\gamma_\pm : \ell^2(\Z) \to \rH_\pm$ defined respectively by the following formulas:
\[
\gamma_+(\Psi) := 
\begin{pmatrix}
\Psi \\ 0
\end{pmatrix} \mbox{ and }
\gamma_-(\Psi) := 
\begin{pmatrix}
0 \\ \Psi
\end{pmatrix},
\qquad \Psi \in \ell^2(\Z).
\]
If we let $Q_{\epsilon_\pm} = \gamma_\mp Q'_{\epsilon_\pm} \gamma_\pm,$ then $\gamma_\mp^* Q_{\epsilon_\pm} = Q'_{\epsilon_\pm} \gamma_\pm.$ More explicitly,
\[
Q'_{\epsilon_+}
\begin{pmatrix}
\Psi \\ 0
\end{pmatrix} = 
\begin{pmatrix}
0 \\ Q_{\epsilon_+}\Psi
\end{pmatrix} \mbox{ and }
Q'_{\epsilon_-}
\begin{pmatrix}
0 \\ \Psi
\end{pmatrix} = 
\begin{pmatrix}
Q_{\epsilon_-}\Psi \\ 0
\end{pmatrix},
\qquad \Psi \in \ell^2(\Z).
\]
With these two equalities in mind, we obtain
\[
Q_\epsilon 
\begin{pmatrix}
\Psi_1 \\ \Psi_2 
\end{pmatrix}
=
Q_\epsilon 
\begin{pmatrix}
\Psi_1 \\ 0
\end{pmatrix}
+
Q_\epsilon 
\begin{pmatrix}
0 \\ \Psi_2 
\end{pmatrix}
=
Q'_{\epsilon_+} 
\begin{pmatrix}
\Psi_1 \\ 0
\end{pmatrix}
+
Q'_{\epsilon_-}
\begin{pmatrix}
0 \\ \Psi_2 
\end{pmatrix}
=
\begin{pmatrix}
Q_{\epsilon_-}\Psi_2 \\ Q_{\epsilon_+}\Psi_1
\end{pmatrix},
\qquad 
\begin{pmatrix}
\Psi_1 \\ \Psi_2 
\end{pmatrix} \in \rH.
\]
Therefore, \cref{Equation: Skew-diagonal Representation of Qepsilon} holds true. Here, the following easy computation shows that the operators $Q'_{\epsilon_\mp}Q'_{\epsilon_\pm}$ and $Q_{\epsilon_\mp}Q_{\epsilon_\pm}$ are unitarily equivalent:
\[
\gamma^*_\pm Q'_{\epsilon_\mp}Q'_{\epsilon_\pm} \gamma_\pm = (\gamma^*_\pm Q'_{\epsilon_\mp} \gamma_\mp) (\gamma^*_\mp Q'_{\epsilon_\pm} \gamma_\pm ) = Q_{\epsilon_\mp}Q_{\epsilon_\pm}.
\]
With this fact in mind, we obtain the following two equalities:
\begin{align*}
&\dim \ker Q'_{\epsilon_\mp}Q'_{\epsilon_\pm} =  \dim \ker Q_{\epsilon_\mp}Q_{\epsilon_\pm}, \\
&\sigma(Q'_{\epsilon_\mp}Q'_{\epsilon_\pm}) \setminus \{0\} = \inf \sigma(Q_{\epsilon_\mp}Q_{\epsilon_\pm}) \setminus \{0\}.
\end{align*}
That is, $(\varGamma_\epsilon,C_\epsilon)$ is Fredholm if and only if $Q_{\epsilon_+}$ is a Fredholm operator and 
\[
\ind(\varGamma_\epsilon,C_\epsilon) = \ind Q'_{\epsilon_+} = \dim \ker Q_{\epsilon_+} -  \dim \ker Q_{\epsilon_-}.
\]
The claim now follows from \cref{Theorem: Invariance of the Witten Index}.
\end{proof}

\subsection{Proof of \cref{Theorem: Witten Index Formula}}

By virtue of \cref{Lemma: Diagonalisation of Shift Implies Off-diagonalisation of Super Charge} we may choose to work with any unitary $\epsilon$ which gives diagonalisation \cref{Equation: Diagonalisation of the Shift Operator} in order to compute the Witten index. In particular, as in \cref{Theorem: Witten Index Formula}, we shall henceforth work with the one given explicitly by \cref{Equation: Definition of Epsilon} in this paper. In order to prove \cref{Theorem: Witten Index Formula}, it remains to show that \cref{Equation: Explicit Definition of Qepsilon} holds true:

\begin{proof}[Proof of Equality \cref{Equation: Explicit Definition of Qepsilon}]
Let $\epsilon$ be the unitary operator given by \cref{Equation: Definition of Epsilon}. We shall first find the matrix representation of the time evolution $U_\epsilon := \epsilon^* U \epsilon.$  We obtain
\begin{equation}
\label{Equation: Uepsilon}
U_\epsilon = 
\epsilon^* U \epsilon
= (\epsilon^* \varGamma \epsilon) (\epsilon^* C \epsilon)
= 
\begin{pmatrix}
1 & 0 \\
0 & -1
\end{pmatrix}
\epsilon^* 
\begin{pmatrix}
a_1 & b^* \\
b & a_2
\end{pmatrix}
\epsilon =: \frac{1}{2}
\begin{pmatrix}
U'_+ & U_- \\
U_+ & U'_- 
\end{pmatrix},
\end{equation}
where \crefrange{Equation1: Unitary Transform by Epsilon}{Equation2: Unitary Transform by Epsilon} allow us to prove:
\begin{align*}
2\epsilon^*
\begin{pmatrix}
a_1 & 0 \\
0 & a_2
\end{pmatrix}
\epsilon &=
\begin{pmatrix}
(1 + p)a_1 + (1-p)a_2(\cdot + 1) & - |q|(a_1 - a_2(\cdot + 1)) \\
- |q|(a_1 - a_2(\cdot + 1)) & (1-p)a_1 + (1 + p)a_2(\cdot + 1)
\end{pmatrix}, \\ 
2\epsilon^*
\begin{pmatrix}
0 & b^* \\
b & 0
\end{pmatrix}
\epsilon &=
\begin{pmatrix}
qLb + q^* b^* L^* & \frac{-(1-p) q Lb + (1 + p) q^* b^*L^*}{|q|} \\
\frac{(1 + p)qLb - (1 - p) q^* b^* L^*}{|q|} & -qLb - q^*b^*L^*
\end{pmatrix},
\end{align*}
where $La_2 = a_2(\cdot + 1) L$ is used in the first equality. Note that \cref{Equation: Uepsilon} becomes
\[
2\epsilon^*
\begin{pmatrix}
a_1 & 0 \\
0 & a_2
\end{pmatrix}
\epsilon +
2\epsilon^*
\begin{pmatrix}
0 & b^* \\
b & 0
\end{pmatrix}
\epsilon
=
2
\epsilon^* 
\begin{pmatrix}
a_1 & b^* \\
b & a_2
\end{pmatrix} \epsilon = 
\begin{pmatrix}
1 & 0 \\
0 & -1
\end{pmatrix}
\begin{pmatrix}
U'_+ & U_- \\
U_+ & U'_- 
\end{pmatrix}
=
\begin{pmatrix}
U'_+ & U_- \\
-U_+ & -U'_- 
\end{pmatrix}.
\]
It can then be shown that the following equalities hold true:
\begin{align}
\label{Equation1: Diagonal Entries of Uepsilon}
U'_\pm & = qLb + q^* b^* L^* \pm (1 \pm p)a_1 \pm (1 \mp p)a_2(\cdot + 1), \\
\label{Equation2: Diagonal Entries of Uepsilon}
U_\pm& = -(1 \pm p) e^{i\theta}Lb + (1 \mp p) e^{-i\theta}b^*L^* \pm |q|(a_1 - a_2(\cdot + 1)) = 2i Q_{\epsilon_\pm}.
\end{align}
We are now in a position to prove \cref{Equation: Explicit Definition of Qepsilon}. We get
\[
2iQ_\epsilon = U_\epsilon -  U^*_\epsilon 
=
\frac{1}{2} 
\begin{pmatrix}
U_+ & U'_- \\
U'_+ & U_- 
\end{pmatrix}
-
\frac{1}{2} 
\begin{pmatrix}
U^*_+ & (U'_+)^* \\
(U'_-)^* & U^*_- 
\end{pmatrix}
 = 
\begin{pmatrix}
0 & \frac{U'_- - (U'_+)^*}{2} \\
\frac{U'_+ - (U'_-)^*}{2}  & 0
\end{pmatrix},
\]
where the last equality follows from the fact that $U_\pm$ given by \cref{Equation1: Diagonal Entries of Uepsilon} are both self-adjoint. On the other hand, $U_\pm$ given by \cref{Equation2: Diagonal Entries of Uepsilon} admit $(U'_\mp)^* = -U'_\pm,$ and so
\[
2iQ_\epsilon
=
\begin{pmatrix}
0 & \frac{U'_- - (U'_+)^*}{2} \\
\frac{U'_+ - (U'_-)^*}{2}  & 0
\end{pmatrix}
=
\begin{pmatrix}
0    & U'_- \\
U'_+ & 0
\end{pmatrix} = 
\begin{pmatrix}
0    & 2iQ_{\epsilon_-} \\
2iQ_{\epsilon_+} & 0
\end{pmatrix}.
\]
The claim follows.
\end{proof}

\subsection{Coin operators with trivial limits}

We shall conclude the current section with one simple corollary of \cref{Theorem: Witten Index Formula}. Recall that in \cref{Theorem: MainTheorem} the case where at least one of the two limits $C(\rL)$ and $C(\rR)$ is a trivial unitary involution is excluded. The following result explains why.

\begin{lemma}
\label{Lemma: Infinite Dimensional Ker Qepsilon}
If $(\varGamma,C)$ is a split-step SUSYQW endowed with an anisotropic coin $C$ with the property that at least one of the two limits $C(\rL)$ and $C(\rR)$ is trivial, then $\dim \ker Q_{\epsilon_\pm} = \infty.$ That is, $(\varGamma,C)$ automatically fails to be Fredholm in this case.
\end{lemma}
\begin{proof}
We may assume without loss of generality that $C(\rL)$ is trivial, and that $C(x) = C(L)$ for each $x \leq 0$ due to the topological invariance \cref{Equation: Topological Invariance}. If $\Psi_\pm \in \ker Q_{\epsilon_\pm},$ then it follows from \cref{Equation: Explicit Definition of Qepsilon} that
\[
0 - 0 \pm |q|(a_2(x+1) - a_1(x)) \Psi_\pm(x) = 0, \qquad x \leq -1,
\]
where $a_2(x+1) - a_1(x) = a_2(\rL) - a_1(\rL) =  0.$ Thus, for each $x \leq -1$ the vectors $\Psi_\pm(x) \in \C^2$ can be freely chosen regardless of the other required conditions $ (Q_{\epsilon_\pm}\Psi_\pm)(x) = 0$ for each $x \geq 0.$ This  implies $\dim \ker Q_{\epsilon_\pm} = \infty.$
\end{proof}

\section{Classification of \texorpdfstring{$\dim \ker Q_{\epsilon_\pm}$}{dimkerQepsilon}}
\label{Section: Classification of Dimensions}

\subsection{The main result}
\label{Subsection: Classification of Dimensions}

In order to state the main theorem of the current section, we introduce the following definition;
\begin{definition}
\label{Assumption: Coin Type}
Let $(\varGamma,C)$ be a split-step SUSYQW with an anisotropic coin $C$. We shall consider the following mutually exclusive cases: 
\begin{align}
\label{Equation: Definition of Type I} \tag{I} b(\rL) = 0 \mbox{ and } b(\rR) = 0,\\
\label{Equation1: Definition of Type II} \tag{II} b(\rL) = 0 \mbox{ and } b(\rR) \neq 0, \\
\label{Equation2: Definition of Type II} \tag{II'} b(\rL) \neq 0 \mbox{ and } b(\rR) = 0, \\
\label{Equation: Definition of Type III} \tag{III} b(\rL) \neq 0 \mbox{ and } b(\rR) \neq 0.
\end{align}
We say that the coin operator $C$ is of \textbi{Type I}, if the two unitary involutions $C(\rL)$ and $C(\rR)$ are both non-trivial and if (I) holds true. Type II, II', III coins are defined likewise. That is, we shall always assume that \cref{Equation: Property of Nontrivial Unitary Involutions} holds true for each $\sharp = \rL,\rR,$ whenever the four types of the isotropic coin thus defined.
\end{definition}

With this definition in mind, the ultimate aim of the current section is to prove the following classification result:
\begin{theorem}
\label{Theorem: Classification of Dimensions}
Let $(\varGamma,C)$ be a split-step SUSYQW, and let $C$ be an anisotropic coin of the following specific form;
\begin{equation}
\label{Equation: Simplified Coin}
C(x) = 
\begin{cases}
C(\rR), & x \geq 1, \\
C(\rL), & x \leq 0,
\end{cases}
\end{equation} 
where $C(\sharp)$ is assumed to be non-trivial for each $\sharp  =\rL,\rR.$ Let $d_\pm = \dim \ker Q_{\epsilon_\pm}.$
\begin{enumerate}
\item If $C$ is of Type \ref{Equation: Definition of Type I}, then $d_\pm$ are uniquely determined by the pair $(a(\rL),a(\rR)):$
\begin{equation}
\label{Equation: Classification of Type I}
d_\pm = 
\begin{cases}
1, & a(\rL)a(\rR) < 0,\\
0, & a(\rL)a(\rR) > 0.
\end{cases}
\end{equation}
\item If $C$ is of Type \ref{Equation1: Definition of Type II}, then $d_\pm$ are uniquely determined by the triple $(p,a(\rL),a(\rR)):$ 
\begin{equation}
\label{Equation1: Classification of Type II}
d_\pm = 
\begin{cases}
1, &\mp p + a(\rL)a(\rR) < 0, \\
0, &\mp p + a(\rL)a(\rR) \geq 0.
\end{cases}
\end{equation}
\item If $C$ is of Type \ref{Equation2: Definition of Type II}, then $d_\pm$ are uniquely determined by the triple $(p,a(\rL),a(\rR)):$ 
\begin{equation}
\label{Equation2: Classification of Type II}
d_\pm = 
\begin{cases}
1, &\pm p + a(\rL)a(\rR) < 0, \\
0, &\pm p + a(\rL)a(\rR) \geq 0, \\
\end{cases}
\end{equation}
\item If $C$ is of Type \ref{Equation: Definition of Type III}, then $d_\pm$ are uniquely determined by the triple $(p,a(\rL),a(\rR)):$
\begin{equation}
\label{Equation: Classification of Type III}
d_\pm = 
\begin{cases}
1, &   a(\rR) < \pm p < a(\rL), \\
1, &  a(\rL) < \mp p  < a(\rR), \\
0, & \mbox{otherwise}.
\end{cases}
\end{equation}
\end{enumerate}
\end{theorem}

\begin{remark}
\label{Remark: Zero Witten Index}
The following comments about \cref{Theorem: Classification of Dimensions} are worth mentioning:
\begin{enumerate}
\item Note that the ultimate purpose of the present paper is not the computation of each individual $d_\pm,$ but rather the difference $d_+ - d_-.$ Since the latter quantity is invariant under compact perturbations, we may impose \cref{Equation: Simplified Coin} without loss of generality.

\item If $p=0,$ then $d_+ = d_-$ regardless of the coin type:
\[
d_\pm = 
\begin{cases}
1, & a(\rL)a(\rR) < 0, \\
0, & a(\rL)a(\rR) \geq 0.
\end{cases}
\]
That is, the Witten index of $(\varGamma,C)$ is always zero in this case.
\end{enumerate}
\end{remark}

\subsection{Preliminaries}

\subsubsection{Notation}
We shall always adhere to the notation introduced here throughout the remaining part of the current section. Let $(\varGamma,C)$ be a split-step SUSYQW endowed with an anisotropic coin $C,$ and let $C(\sharp)$ be non-trivial for each $\sharp = \rL,\rR.$ Recall that $Q_{\epsilon_\pm}$ introduced in \cref{Theorem: Witten Index Formula} are operators of the following forms:
\[
Q_{\epsilon_\pm} = (1 \pm p) e^{i \theta}b(\cdot + 1)L - (1 \mp p)e^{-i \theta}b^*L^* \pm |q|(a_2(\cdot + 1) - a_1),
\]
where the unnecessary constant $-2i$ is removed for notational simplicity.  We have
\begin{align}
\label{Equation: Definition of Qepsilon}
Q_{\epsilon_\pm} &= \alpha_\pm(\cdot + 1)L- \alpha^*_\mp L^* \pm \beta, \\
\alpha_\pm &:= (1 \pm p) e^{i \theta}b, \\
\beta &:= |q|(a_2(\cdot + 1) - a_1),
\end{align}
where the two-sided limits of the last two sequences will be denoted respectively by 
\begin{align*}
\alpha_\pm(\sharp) &:= \lim_{x \to \sharp} \alpha_\pm(x) = (1 \pm p) e^{i \theta}b(\sharp),& &\qquad \sharp = \rL,\rR, \\
\beta(\sharp) &:= \lim_{x \to \sharp} \beta(x) = |q|(a_2(\sharp) - a_1(\sharp)) = -2|q|a(\sharp),& &\qquad \sharp = \rL,\rR,
\end{align*}
where the last equality follows from \cref{Equation: Property of Nontrivial Unitary Involutions}. We shall also make use of the simplification assumption \cref{Equation: Simplified Coin} throughout this subsection, so that
\begin{align} 
\alpha_\pm(x) &= 
\begin{cases}
\alpha_\pm(\rR), & x \geq 1, \\
\alpha_\pm(\rL), & x \leq 0,
\end{cases} \\
\beta(x) &= 
\begin{cases}
\beta(\rR), & x \geq 1, \\
-|q|(a(\rL) + a(\rR)), & x = 0,\\
\beta(\rL), & x \leq -1. 
\end{cases}
\end{align}

\subsubsection{A sketch for the proof of \cref{Theorem: Classification of Dimensions}}
\label{Section: Sketch}
The main theorem of the current section, \cref{Theorem: Classification of Dimensions}, does require a lengthy argument as we need to separately consider the four types of the coin operator. However, the basic idea behind the proof is in fact elementary. Note first that the equation $(Q_{\epsilon_\pm}\Psi)(x) = 0$ is equivalent to 
\begin{equation}
\label{Equation: Motivation about DE}
\alpha_\pm(x + 1)\Psi(x+1) - \alpha_\mp(x)^* \Psi(x-1)\pm \beta(x)\Psi(x) = 0.
\end{equation}
This equation, known as the second-order linear difference equation, can then be put into the following first-order matrix equation;
\begin{equation}
\label{Equation: Motivation about Matrix DE}
\begin{pmatrix}
\Psi(x+1) \\
\Psi(x)
\end{pmatrix}
=  
\begin{pmatrix}
\frac{\mp \beta(x)}{\alpha_\pm(x + 1)} & \frac{\alpha_\mp(x)^*}{\alpha_\pm(x + 1)} \\
1 & 0
\end{pmatrix}
\begin{pmatrix}
\Psi(x) \\ \Psi(x-1)
\end{pmatrix},
\end{equation}
whenever $\alpha_\pm(x + 1) \neq 0.$ This idea of transforming a difference equation to the associated matrix equation of less order is well-known (see, for example, \cite{Book:Elaydi05:AnIntroductionToDifferenceEquations}), and this is precisely the approach we are going to take. Note that we need not only to algebraically solve Equation \cref{Equation: Motivation about DE}, but also to ensure the solutions to be square summable. The following coefficient matrix shall be used throughout this section;
\begin{align}
\label{Equation: Definition of Apm}
A_\pm(x) &:=
\begin{pmatrix} 
\frac{\mp \beta(x)}{\alpha_\pm(x+1)} &  \frac{\alpha_\mp(x)^*}{\alpha_\pm(x+1)} \\
1 & 0
\end{pmatrix}, \\
\label{Equation: Definition of Apm Sharp}
A_\pm(\sharp) &:= \lim_{x \to \sharp} A_\pm(x) = 
\begin{pmatrix} 
\frac{\mp \beta(\sharp)}{\alpha_\pm(\sharp)} &  \frac{\alpha_\mp(\sharp)^*}{\alpha_\pm(\sharp)} \\
1 & 0
\end{pmatrix}, \qquad \sharp = \rL,\rR,
\end{align}
where $A_\pm(\rL)$ are well-defined if $C$ is of either Type \ref{Equation1: Definition of Type II} or \ref{Equation: Definition of Type III}, and $A_\pm(\rR)$ are well-defined if $C$ is of either Type \ref{Equation2: Definition of Type II} or \ref{Equation: Definition of Type III}.

\subsubsection{Abstract matrix difference equations}
We are interested in solving a \textbi{first-order linear matrix difference equation} which is an equation of the following form:
\begin{equation}
\label{Equation: General Matrix Difference Equation} 
\Phi(x+1) = A_0 \Phi(x), \qquad x \in \N,
\end{equation}
where $\N = \{1,2,\dots\}$ and $A_0$ is a fixed invertible $2 \times 2$ matrix. An easy inductive argument shows that \cref{Equation: General Matrix Difference Equation} is equivalent to the following equation:
\begin{equation}
\label{Equation: Inductive Easier Matrix Equation} 
\Phi(x+1) = \underbrace{A_0 \dots A_0}_{\mbox{$x$ times}} \Phi(1) = A_0^x \Phi(1), \qquad x \geq 0.
\end{equation}

We call any $\C^2$-valued sequence $\Phi$ satisfying \cref{Equation: General Matrix Difference Equation} an \textbi{algebraic solution} with in mind that it may fail to be square summable. It is easy to see from \cref{Equation: Inductive Easier Matrix Equation} that any algebraic solution $\Phi$ is uniquely determined by the initial value $\Phi(1).$ Given a $\C^2$-valued sequence $\Phi,$ we have that $\Phi$ is an algebraic solution to \cref{Equation: General Matrix Difference Equation} and $\sum_{x \in \N} \|\Phi(x)\|^2 < \infty$ if and only if $\Phi \in \ker(L \oplus L - \bigoplus_{x \in \Z} A_0).$ The following well-known result is included merely for the sake of completeness (See, for example, the proof of \cite[Theorem 2.15]{Book:Elaydi05:AnIntroductionToDifferenceEquations} which makes use of the discrete analogue of the Wronskian);

\begin{lemma}
\label{Lemma: Linear Independence of Solutions to Matrix Difference Equation}
Let $A_0$ be a fixed invertible $2 \times 2$ matrix, and let $\Phi,\Phi'$ be two algebraic solutions to the difference equation \cref{Equation: General Matrix Difference Equation}. Then $\Phi,\Phi'$ are linearly independent if and only if $\Phi(x_0),\Phi'(x_0)$ are linearly independent for any $x_0 \in \N.$
\end{lemma}
\begin{proof}
If $\Phi(x_0),\Phi'(x_0)$ are linearly independent for each $x_0 \in \Z,$ then $\Phi,\Phi'$ are obviously linearly independent. To prove the converse, suppose that $\Phi,\Phi'$ are linearly independent, and that $c\Phi(x_0) + c'\Phi'(x_0) = 0$ for some fixed $x_0 \in \N$ and some $c,c' \in \C.$ Since $\Phi,\Phi'$ are both solutions to the difference equation \cref{Equation: General Matrix Difference Equation}, we have $c\Phi(x) + c'\Phi'(x) = 0$ for each $x \in \N,$ and so the linear independence of $\Phi,\Phi'$ gives $c = c' = 0.$ It follows that $\Phi(x_0),\Phi'(x_0)$ are linearly independent for each $x_0 \in \N.$ 
\end{proof}

To put it another way, \cref{Lemma: Linear Independence of Solutions to Matrix Difference Equation} states that two algebraic solutions $\Phi,\Phi'$ to \cref{Equation: General Matrix Difference Equation} are either identically linearly independent or identically linearly dependent. 
It is then easy to observe that $\dim \ker(L \oplus L - \bigoplus_{x \in \N} A_0) \leq 2.$ Here, the equality may not hold, since an algebraic solution to \cref{Equation: Inductive Easier Matrix Equation} may fail to be square summable. To check the square summability of solutions, the following lemma is useful:
\begin{lemma}
\label{Lemma: Square Summability}
Let $A_0$ be a fixed invertible $2 \times 2$ matrix with two distinct eigenvalues $z_1,z_2,$ so that $A_0$ admits diagonalisation of the following form for some invertible matrix $P:$
\[
A_0 = 
P
\begin{pmatrix}
z_1 & 0 \\
0 & z_2
\end{pmatrix}P^{-1}.
\]
Suppose that $\Phi$ is an algebraic solution to the difference equation \cref{Equation: General Matrix Difference Equation}, and that
\begin{equation}
\label{Equation: Definition of k}
\begin{pmatrix}
k_{1} \\
k_{2} \\
\end{pmatrix}
:= P^{-1} \Phi(1). 
\end{equation}
Then we have $\Phi \in \ker(L \oplus L - \bigoplus_{x \in \N} A_0)$ if and only if the following sum is finite; 
\[
\sum_{x \in \N} \left(|k_{1}|^2 |z_{1}|^{2x} + |k_{2}|^2 |z_{2}|^{2x}\right) < \infty.
\]
\end{lemma}
\begin{proof}
It follows from \cref{Equation: Inductive Easier Matrix Equation} that
\[
\Phi(x+1) = A_0^x \Phi(1) 
= 
P 
\begin{pmatrix}
z_1^x & 0 \\
0 & z_2^x
\end{pmatrix}
P^{-1}\Phi(1) =
P 
\begin{pmatrix}
k_1z_1^x  \\
k_2z_2^x
\end{pmatrix},\qquad x \geq 0.
\]
Then there exist constants $C_1,C_2 > 0,$ such that
\[
C_1 
\left\|
\begin{pmatrix}
k_1 z_1^x\\
k_2 z_2^x
\end{pmatrix}
\right\|^2 \leq 
\|\Phi(x + 1)\|^2 \leq C_2 
\left\|
\begin{pmatrix}
k_1 z_1^x\\
k_2 z_2^x
\end{pmatrix}
\right\|^2,
\qquad x \geq 0,
\]
where
\[
\left\|
\begin{pmatrix}
k_1 z_1^x\\
k_2 z_2^x
\end{pmatrix}
\right
\|^2  
= |k_1|^2 |z_1|^{2x} + |k_2|^2 |z_2|^{2x}.
\]
The claim follows.
\end{proof}

In fact, we shall end up solving Equation \cref{Equation: General Matrix Difference Equation} with a constraint on the initial condition, and so we introduce the following notation:

\begin{lemma}
Given an invertible $2 \times 2$ matrix $A_0$ and two complex numbers $a_0,b_0 \in \C,$ we introduce the following subspace of $\ell^2(\N,\C^2):$
\begin{equation}
\label{Equation: Difference Equation with Constraints}
\eS(A_0,a_0,b_0) := 
\left\{
	\begin{pmatrix}\Phi_1 \\ \Phi_2\end{pmatrix} \in \ker \left(L \oplus L - \bigoplus_{x \in \N} A_0 \right) \mid a_0\Phi_1(1) + b_0\Phi_2(1) = 0
\right\}.
\end{equation}
If $a_0,b_0 \in \C$ are both non-zero, then $\dim \eS(A_0,a_0,b_0) \leq 1.$
\end{lemma}
\begin{proof}
If $\Phi,\Phi' \in \eS(A_0,a_0,b_0),$ then $\Phi(1),\Phi'(1)$ are linearly dependent vectors in $\C^2:$
\[
a_0\Phi(1) + b_0\Phi'(1) =
a_0\begin{pmatrix}\Phi_1(1) \\ \Phi_2(1)\end{pmatrix} + 
b_0\begin{pmatrix}\Phi'_1(1) \\ \Phi'_2(1)\end{pmatrix}
=
\begin{pmatrix}
a_0\Phi_1(1) + b_0\Phi_2(1)\\
a_0\Phi'_1(1) + b_0\Phi'_2(1)
\end{pmatrix}
=0.
\]
Thus, $\Phi,\Phi'$ are also linearly dependent by \cref{Lemma: Linear Independence of Solutions to Matrix Difference Equation}. The claim follows.
\end{proof}

\subsubsection{Concrete matrix difference equations}
With \cref{Equation: Difference Equation with Constraints} in mind, we are now in a position to state the explicit forms of the matrix difference equations we need to solve;
\begin{lemma}
If the type of the coin $C$ is one of II,II',III, then we have the following associated linear isomorphisms respectively:
\begin{align}
\label{Equation1: Isomorphism from Qepsilon}
&\ker Q_{\epsilon_\pm} \ni (\Psi(x))_{x \in \Z} \longmapsto 
\left( 
\begin{pmatrix}
\Psi(x) \\
\Psi(x-1)
\end{pmatrix}\right)_{x \in \N} \in  
\eS(A_\pm(\rR),\alpha_\pm(\rR), \pm \beta(0)), \\
\label{Equation2: Isomorphism from Qepsilon}
&\ker Q_{\epsilon_\pm} \ni  (\Psi(x))_{x \in \Z} \longmapsto 
\left( 
\begin{pmatrix}
\Psi(-(x-1)) \\
\Psi(-x)
\end{pmatrix}\right)_{x \in \N}  \in 
\eS(A_\pm(\rL)^{-1},\mp \beta(0),\alpha_\mp(\rL)^*),  \\
\label{Equation3: Isomorphism from Qepsilon}
&\ker Q_{\epsilon_\pm} \ni \Psi \longmapsto 
\left(
\begin{pmatrix}
\Psi(x) \\
\Psi(x-1)
\end{pmatrix}\right)_{x \in \Z}\in  \ker\left(L \oplus L - \bigoplus_{x \in \Z} A_\pm(x)\right).
\end{align}
\end{lemma}
\begin{proof}
If $C$ is of Type \ref{Equation1: Definition of Type II}, then $\alpha_\pm(\rL) = 0$ and $\beta(\rL) \neq 0.$ Thus $\Psi \in \ker Q_{\epsilon_\pm}$ if and only if \cref{Equation: Motivation about DE} holds true for each $x \geq 1$ together with the following two conditions:
\[
\alpha_\pm(\rR)\Psi(1) \pm \beta(0)\Psi(0) = 0 \mbox{ and } \Psi(x) = 0, 
\quad x \leq -1.
\]
As in \cref{Section: Sketch}, Equation \cref{Equation: Motivation about DE} is equivalent to the following:
\[
L \oplus L
\begin{pmatrix}
\Psi(x) \\ \Psi(x-1)
\end{pmatrix}
=  
\begin{pmatrix} 
\frac{\mp \beta(x)}{\alpha_\pm(x+1)} &  \frac{\alpha_\mp(x)^*}{\alpha_\pm(x+1)} \\
1 & 0
\end{pmatrix}
\begin{pmatrix}
\Psi(x) \\ \Psi(x-1)
\end{pmatrix}
=
A_\pm(\rR)
\begin{pmatrix}
\Psi(x) \\ \Psi(x-1)
\end{pmatrix},
\qquad x \geq 1,
\]
and so \cref{Equation1: Isomorphism from Qepsilon} is a well-defined operator. It remains to show that \cref{Equation1: Isomorphism from Qepsilon} is surjective, since the injectivity is obvious. It is easy to verify that any vector in $\eS(A_\pm(\rR),\alpha_\pm(\rR), \pm \beta(0))$ must be of the form $(L\Psi_0,\Psi_0)$ for some $\Psi_0 \in \ell^2(\N)$ satisfying 
\[
L \oplus L
\begin{pmatrix}
\Psi_0(x+1) \\ \Psi_0(x)
\end{pmatrix}
=
A_\pm(\rR)
\begin{pmatrix}
\Psi_0(x+1) \\ \Psi_0(x)
\end{pmatrix},
\qquad  x \geq 1,
\]
We define $\Psi \in \ell^2(\Z,\C^2)$ by
\[
\Psi(x-1) := 
\begin{cases}
0, & x \leq 0, \\
\Psi_0(x), & x \geq 1.
\end{cases} 
\]
Then it is easy to show that $\Psi \in \ker Q_{\epsilon_\pm},$ and that it gets mapped to $(L\Psi_0,\Psi_0)$ under \cref{Equation1: Isomorphism from Qepsilon}. Therefore, \cref{Equation1: Isomorphism from Qepsilon} is a well-defined linear isomorphism.

Similarly, if $C$ is of Type \ref{Equation2: Definition of Type II}, then $\alpha_\pm(\rR) = 0$ and $\beta(\rR) \neq 0.$ Thus $\Psi \in \ker Q_{\epsilon_\pm}$ if and only if \cref{Equation: Motivation about DE} holds true for each $x \leq -1$ together with the following two conditions:
\[
\mp \beta(0)\Psi(0) + \alpha_\mp(\rL)^*\Psi(-1)  = 0 \mbox{ and } \Psi(x) = 0, \quad x \geq +1.
\]
As before \cref{Equation: Motivation about DE} is equivalent to the following:
\[
\begin{pmatrix}
\Psi(x+1) \\ \Psi(x)
\end{pmatrix}
=
A_\pm(\rL)
\begin{pmatrix}
\Psi(x) \\ \Psi(x-1)
\end{pmatrix},
\qquad x \leq -1.
\]
If we introduce the change of variable $x \leftrightarrow -x,$ then the above equation becomes
\[
L \oplus L
\begin{pmatrix}
\Psi(-(x-1)) \\ \Psi(-x)
\end{pmatrix}
=
\begin{pmatrix}
\Psi(-x) \\ \Psi(-x-1)
\end{pmatrix} = 
A_\pm(\rL)^{-1}
\begin{pmatrix}
\Psi(-(x-1)) \\ \Psi(-x)
\end{pmatrix},
\qquad  x \geq 1.
\]
and so \cref{Equation2: Isomorphism from Qepsilon} is a well-defined operator. It remains to show that \cref{Equation2: Isomorphism from Qepsilon} is surjective, since the injectivity is obvious. It is easy to verify that any vector in $\eS(A_\pm(\rL)^{-1},\mp \beta(0),\alpha_\mp(\rL)^*)$ must be of the form $(\Psi_0,L\Psi_0)$ for some $\Psi_0 \in \ell^2(\N)$ with
\[
\begin{pmatrix}
\Psi_0(x) \\ \Psi_0(x+1)
\end{pmatrix}
=
A_\pm(\rL)
\begin{pmatrix}
\Psi_0(x+1) \\ \Psi_0(x+2)
\end{pmatrix},
\qquad  x \geq 1.
\]
We define $\Psi \in \ell^2(\Z,\C^2)$ by
\[
\Psi(-(x-1)) := 
\begin{cases}
0, & x \leq 0, \\
\Psi_0(x), & x \geq 1.
\end{cases} 
\]
Then it is easy to show that $\Psi \in \ker Q_{\epsilon_\pm},$ and that it gets mapped to $(\Psi_0,L\Psi_0)$ under \cref{Equation2: Isomorphism from Qepsilon}. Therefore, \cref{Equation2: Isomorphism from Qepsilon} is a well-defined linear isomorphism. The fact that \cref{Equation3: Isomorphism from Qepsilon} is a linear isomorphism if $C$ is of Type \ref{Equation: Definition of Type III} is left as an easy exercise.
\end{proof}

As in \cref{Lemma: Square Summability}, diagonalisation of the coefficient matrices $A_\pm(\sharp)$ is important.
\begin{lemma}
\label{Lemma: Eigenvalues of Apm}
If $A_\pm(\sharp)$ are well-defined for $\sharp = \rL,\rR,$ then the two matrices $A_\pm(\sharp)$ have two non-zero distinct eigenvalues $z_{\pm,1}(\sharp), z_{\pm,2}(\sharp)$ of the following forms:
\begin{equation}
\label{Equation: Eigenvalues of Apmx}
z_{\pm,j}(\sharp) = \frac{q^*}{1 \pm p} \left(\frac{(-1)^j \pm a(\sharp)}{b(\sharp)}\right),
\qquad  j=1,2.
\end{equation}
Moreover, the two matrices $A_\pm(\sharp)$ admit diagonalisation of the following form;
\begin{align}
A_\pm(\sharp) &= 
P_\pm(\sharp)
\begin{pmatrix}
z_{\pm,1}(\sharp) & 0 \\
0 & z_{\pm,2}(\sharp)
\end{pmatrix}
P_\pm(\sharp)^{-1}, \\
\label{Equation: Definition of Ppm}
P_\pm(\sharp) &:= 
\begin{pmatrix}
z_{\pm,1}(\sharp) & z_{\pm,2}(\sharp) \\
1 & 1
\end{pmatrix}.
\end{align}
\end{lemma}
\begin{proof}
It is left as an easy exercise to show that a $2 \times 2$ matrix of the form
\[
\begin{pmatrix}
s & t \\
1 & 0
\end{pmatrix}
\]
has two eigenvalues $2z_j = 
s + (-1)^j \sqrt{s^2 + 4t},$ where $j=1,2,$ together with the following eigenvalue equations:
\[
\begin{pmatrix}
s & t \\
1 & 0
\end{pmatrix}
\begin{pmatrix}
z_j \\
1
\end{pmatrix} = 
z_j
\begin{pmatrix}
z_j \\
1
\end{pmatrix},
\qquad  j=1,2.
\]
With this result in mind, the matrices $A_\pm(\sharp)$ given by \cref{Equation: Definition of Apm Sharp} are as follows;
\[
2z_{\pm,j}(\sharp) 
= \frac{\mp \beta(\sharp) + (-1)^j \sqrt{\beta(\sharp)^2 + 4 \alpha_\pm(\sharp) \alpha_\mp(\sharp)^*} }{\alpha_\pm(\sharp)},
\qquad j = 1,2,
\]
where 
\[
\beta(\sharp)^2 + 4 \alpha_\pm(\sharp) \alpha_\mp(\sharp)^*  
= 4|q|^2 a(\sharp)^2 + 4(1 - p^2)|b(\sharp)|^2 = 4|q|^2 > 0.
\]
It follows that
\[
z_{\pm,j}(\sharp) 
= \frac{\mp \beta(\sharp) + (-1)^j \sqrt{\beta(\sharp)^2 + 4 \alpha_\pm(\sharp) \alpha_\mp(\sharp)^*} }{2 \alpha_\pm(\sharp)}
= \frac{q^*}{1 \pm p} \left(\frac{(-1)^j \pm a(\sharp)}{b(\sharp)}\right).
\]
The claim follows.
\end{proof}

\begin{definition}
We define the increasing function $f : [-1,1] \to [0,+\infty]$ by
\[
f(\kappa) := 
\begin{cases}
\sqrt{\frac{1 + \kappa}{1 - \kappa}}, & \kappa \neq 1 \\
+\infty, & \kappa = 1.
\end{cases}
\]
\end{definition}

The following figure shows the graphs of $y = f(\pm \kappa):$
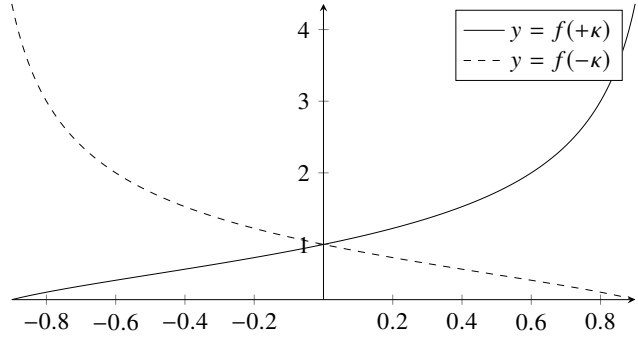
\[
\begin{tikzpicture}[scale=0.9]
\begin{axis}[axis lines=center,width = 0.9\textwidth,height = 0.5\textwidth]
\addplot[samples=200,domain=-0.9:0.9]{sqrt((1 + x)/(1 - x))};
\addlegendentry{$y = f(+\kappa)$}
\addplot[dashed,samples=200,domain=-0.9:0.9]{sqrt((1 - x)/(1 + x))};
\addlegendentry{$y = f(-\kappa)$}
\end{axis}
\end{tikzpicture}
\]
We shall also make use of the following obvious identities:
\begin{align}
\label{Equation: Reflection Property}
&f(\kappa) ^{-1} = f(-\kappa), \\
\label{Equation: Multiplicative Property}
&f(\kappa)f(\kappa')  = f\left(\frac{\kappa + \kappa'}{1 + \kappa \kappa'}\right), \\
\label{Equation: Additivity and Multiplicative}
&f(\kappa)f(\kappa') < 1 \mbox{ if and only if } \kappa + \kappa' < 0,
\end{align}
where $\kappa,\kappa' \in (-1,1).$

\begin{corollary}
With the notation introduced in \cref{Lemma: Eigenvalues of Apm} in mind, we have
\begin{equation}
\label{Equation: Absolute Value of Zpm}
|z_{\pm,j}(\sharp)| = 
\begin{cases}
f(\mp p)f(\mp a(\sharp)), & j=1, \\
f(\mp p)f(\pm a(\sharp)), & j=2.
\end{cases}
\end{equation}
\end{corollary}
\begin{proof} 
Since $qq^* = (1-p)(1+p)$ and $b(\sharp)b(\sharp)^* = ((-1)^j - a(\sharp))((-1)^j + a(\sharp)),$ 
\begin{equation}
\label{Equation: Little Trick}
\frac{q^*}{1 \pm p} = \frac{1 \mp p}{q} \mbox{ and }
\frac{b(\sharp)^*}{(-1)^j \pm a(\sharp)} = \frac{(-1)^j \mp a(\sharp)}{b(\sharp)}.
\end{equation}

where 
\[
\left| \frac{q^*}{1 \pm p}\right| = 
\frac{|q|}{|1 \pm p|} = 
\frac{\sqrt{(1+p)(1-p)}}{\sqrt{(1 \pm p)^2}} = 
\sqrt{\frac{1 \mp p}{1 \pm p}} = f(\mp p).
\]
On the other hand,
\[
\left|\frac{(-1)^j \pm  a(\sharp) }{b(\sharp)}\right| =
\frac{|(-1)^j \pm  a(\sharp) |}{|b(\sharp)|}
= \frac{\sqrt{((-1)^j \pm  a(\sharp))^2}}{\sqrt{(1 + a(\sharp))(1 - a(\sharp))}}
= \sqrt{\frac{((-1)^j \pm  a(\sharp))^2}{(1 + a(\sharp))(1 - a(\sharp))}}.
\]
We get
\begin{align*}
\left|\frac{(-1)^1 \pm  a(\sharp) }{b(\sharp)}\right| &= \sqrt{\frac{(-1 \pm  a(\sharp))^2}{(1 + a(\sharp))(1 - a(\sharp))}} = \sqrt{\frac{(1 \mp  a(\sharp))^2}{(1 + a(\sharp))(1 - a(\sharp))}} = \sqrt{\frac{1 \mp  a(\sharp)}{1 \pm a(\sharp)}}, \\
\left|\frac{(-1)^2 \pm  a(\sharp) }{b(\sharp)}\right| &= \sqrt{\frac{(1 \pm  a(\sharp))^2}{(1 + a(\sharp))(1 - a(\sharp))}} = \sqrt{\frac{1 \pm  a(\sharp)}{1 \mp a(\sharp)}}.
\end{align*}
The claim follows.
\end{proof}

\subsection{Proof of \cref{Theorem: Classification of Dimensions}}

It remains to prove \crefrange{Equation: Classification of Type I}{Equation: Classification of Type III}.

\subsubsection{Type \ref{Equation: Definition of Type I} coin}
\begin{proof}[Proof of Equality \cref{Equation: Classification of Type I}]
If $C$ is an anisotropic coin of Type \ref{Equation: Definition of Type I}, then $\beta(\sharp) \neq 0$ for each $\sharp = \rL,\rR.$ It follows that $\Psi \in \ker Q_{\epsilon_\pm}$ if and only if $\Psi(x) = 0$ whenever $x \neq 0.$ We get
\[
\beta(0) = -|q|(a(\rR) + a(\rL)) = 
\begin{cases}
0, & a(\rL) a(\rR) < 0, \\
\mbox{non-zero}, & a(\rL)a(\rR) > 0.
\end{cases}
\]
The claim follows.
\end{proof}

\subsubsection{Type \ref{Equation1: Definition of Type II} coin} 
If $C$ is an anisotropic coin of Type \ref{Equation1: Definition of Type II}, then we shall make use of the isomorphism \cref{Equation1: Isomorphism from Qepsilon};
\[
d_\pm = \dim Q_{\epsilon_\pm} = \dim \eS(A_\pm(\rR),\alpha_\pm(\rR), \pm \beta(0)) \leq 1.
\]
We shall compute $d_\pm$ by making use of \cref{Lemma: Square Summability}. As in \cref{Lemma: Eigenvalues of Apm}, the matrices $A_\pm(\rR)$ admit diagonalisation of the following form:
\begin{align*}
A_\pm(\rR) &= P_\pm(\rR)
\begin{pmatrix}
z_{\pm,1}(\rR) & 0 \\
0 & z_{\pm,2}(\rR)
\end{pmatrix} P_\pm(\rR)^{-1}, \\
P_\pm(\rR) &= 
\begin{pmatrix}
z_{\pm,1}(\rR) & z_{\pm,2}(\rR) \\
1 & 1
\end{pmatrix}.
\end{align*}
Given $\C$-valued sequences $\Phi_\pm = (\Phi_\pm(x))_{x \in \N},$ we have $\Phi_\pm \in \eS(A_\pm(\rR),\alpha_\pm(\rR), \pm \beta(0))$ if and only if $\Phi_\pm$ are square-summable and the following equalities hold true:
\begin{align}
\label{Equation1: One-sided DE}&\Phi_\pm(x + 1) = A_\pm(\rR)\Phi_\pm(x), &&x \in \N, \\ 
\label{Equation2: One-sided DE} &\Phi_\pm(1) = m_\pm
\begin{pmatrix}
\frac{\mp \beta(0)}{\alpha_\pm(\rR)} \\
1
\end{pmatrix},
&&\exists m_\pm \in \C.
\end{align}

\begin{lemma}
If the sequences $\Phi_\pm$ satisfy \cref{Equation1: One-sided DE} and \cref{Equation2: One-sided DE}, then
\begin{equation}
\label{Equation1: Definition of kpm}
\begin{pmatrix}
k_{\pm,1} \\
k_{\pm,2}
\end{pmatrix} :=
P_\pm(\rR)^{-1} \Phi_\pm(1) = 
\frac{-m_\pm}{2}
\begin{pmatrix}
-1 \pm a(\rL)\\
-1 \mp a(\rL)
\end{pmatrix}.
\end{equation}
\end{lemma}
\begin{proof}
If we let $\sharp = \rR,$ then
\begin{equation} 
\label{Equation: To Be Used Later}
\begin{pmatrix}
k_{\pm,1} \\
k_{\pm,2}
\end{pmatrix} =
\begin{pmatrix}
z_{\pm,1}(\sharp) & z_{\pm,2}(\sharp) \\
1 & 1
\end{pmatrix}^{-1} 
\Phi_\pm(1) = 
\frac{1}{\det P_\pm(\sharp)}
\begin{pmatrix}
1 & -z_{\pm,2}(\sharp) \\
-1 & z_{\pm,1}(\sharp)
\end{pmatrix}
\Phi_\pm(1),
\end{equation}
where \cref{Lemma: Eigenvalues of Apm} implies
\[
\det P_\pm(\sharp) = z_{\pm,1}(\sharp) - z_{\pm,2}(\sharp) = \frac{q^*}{1 \pm p} \left(\frac{-1 \pm a(\sharp)}{b(\sharp)}\right) - 
\frac{q^*}{1 \pm p} \left(\frac{+1 \pm a(\sharp)}{b(\sharp)}\right)
= \frac{-2q^*}{b(\sharp)(1 \pm p)}.
\]
With this equality in mind we obtain
\[
\frac{\det P_\pm(\rR)}{m_\pm}
\begin{pmatrix}
k_{\pm,1} \\
k_{\pm,2}
\end{pmatrix}
= 
\begin{pmatrix}
1 & -z_{\pm,2}(\rR) \\
-1 & z_{\pm,1}(\rR)
\end{pmatrix}
\begin{pmatrix}
\frac{\mp \beta(0)}{\alpha_\pm(\rR)} \\
1
\end{pmatrix}
=
\begin{pmatrix}
-\left(\frac{\pm \beta(0)}{\alpha_\pm(\rR)}  +z_{\pm,2}(\rR) \right)\\
\frac{\pm \beta(0)}{\alpha_\pm(\rR)}  +z_{\pm,1}(\rR)
\end{pmatrix},
\]
where
\[
\frac{\pm \beta(0)}{\alpha_\pm(\rR)} + z_{\pm,j}(\rR) = 
\frac{q^*(\mp a(\rL) \mp a(\rR))}{(1 \pm p)  b(\rR)}  + 
\frac{q^*((-1)^j \pm a(\rR))}{(1 \pm p)  b(\rR)}
= \frac{q^*((-1)^j \mp a(\rL))}{(1 \pm p)  b(\rR)}
\]
Therefore
\[
\frac{\det P_\pm(\rR)}{m_\pm}
\begin{pmatrix}
k_{\pm,1} \\
k_{\pm,2}
\end{pmatrix}
= 
\frac{q^*}{(1 \pm p)  b(\rR)} 
\begin{pmatrix}
-1 \pm a(\rL)\\
-1 \mp a(\rL)
\end{pmatrix} = 
\frac{-\det P_\pm(\rR)}{2}
\begin{pmatrix}
-1 \pm a(\rL)\\
-1 \mp a(\rL)
\end{pmatrix}.
\]
\end{proof}

\begin{proof}[Proof of Equality \cref{Equation1: Classification of Type II}]
We shall first assume $a(\rL) = 1.$ If the sequences $\Phi_\pm$ satisfy \cref{Equation1: One-sided DE} and \cref{Equation2: One-sided DE}, then
\[
\begin{pmatrix}
k_{+,1} \\
k_{+,2}
\end{pmatrix} =
\begin{pmatrix}
0 \\
m_+ 
\end{pmatrix} \mbox{ and }
\begin{pmatrix}
k_{-,1} \\
k_{-,2}
\end{pmatrix} = 
\begin{pmatrix}
m_- \\
 0
\end{pmatrix}.
\]
Thus $\Phi_+$ is square summable (resp. $\Phi_-$ is square summable) if and only if
\[
|m_+|^2 \sum_{x=0}^\infty |z_{+,2}(\rR)|^{2x} < \infty \qquad \left(\mbox{resp. } |m_-|^2 \sum_{x=0}^\infty |z_{-,1}(\rR)|^{2x} < \infty \right),
\]
where \cref{Equation: Absolute Value of Zpm} gives
\[
|z_{+,2}(\rR)| = f(- p)f(+ a(\rR)) \mbox{ and } |z_{-,1}(\rR)| = f(+ p)f(+ a(\rR)).
\]
Therefore, we obtain
\[
d_\pm = 
\begin{cases}
1, &\mp p + a(\rR) < 0, \\
0, &\mp p + a(\rR) \geq 0. \\
\end{cases}
\]
An analogous argument gives that if $a(\rL) = -1,$ then
\[
d_\pm = 
\begin{cases}
1, &\mp p - a(\rR) < 0, \\
0, &\mp p - a(\rR) \geq 0. \\
\end{cases}
\]
Thus, \cref{Equation1: Classification of Type II} is proved. 
\end{proof}

\subsubsection{Type \ref{Equation2: Definition of Type II} coin} 
This case is nothing but a repetition of the previous argument, but we include the proof for completeness. If $C$ is an anisotropic coin of Type \ref{Equation2: Definition of Type II}, then we shall make use of the isomorphism \cref{Equation2: Isomorphism from Qepsilon};
\[
d_\pm = \dim Q_{\epsilon_\pm} = \dim \eS(A_\pm(\rL)^{-1},\mp \beta(0),\alpha_\mp(\rL)^*) \leq 1.
\]
We shall compute $d_\pm$ by making use of \cref{Lemma: Square Summability}. As in \cref{Lemma: Eigenvalues of Apm}, the matrices $A_\pm(\rL)^{-1}$ admit diagonalisation of the following form:
\begin{align*}
A_\pm(\rL)^{-1} &= P_\pm(\rL)
\begin{pmatrix}
z_{\pm,1}(\rL)^{-1} & 0 \\
0 & z_{\pm,2}(\rL)^{-1}
\end{pmatrix} P_\pm(\rL)^{-1}, \\
P_\pm(\rL) &= 
\begin{pmatrix}
z_{\pm,1}(\rL) & z_{\pm,2}(\rL) \\
1 & 1
\end{pmatrix}.
\end{align*}
Given $\C$-valued sequences $\Phi_\pm = (\Phi_\pm(x))_{x \in \N},$ we get $\Phi_\pm \in \eS(A_\pm(\rL)^{-1},\mp \beta(0),\alpha_\mp(\rL)^*)$ if and only if $\Phi_\pm$ are square-summable and the following algebraic conditions hold:
\begin{align}
\label{Equation3: One-sided DE}&\Phi_\pm(x + 1) = A_\pm(\rL)^{-1}\Phi_\pm(x), &&x \in \N, \\ 
\label{Equation4: One-sided DE} &\Phi_\pm(1) = m_\pm
\begin{pmatrix}
1\\
\frac{\pm \beta(0)}{\alpha_\mp(\rL)^*}
\end{pmatrix}, 
&&\exists m_\pm \in \C, 
\end{align}

\begin{lemma}
If the sequences $\Phi_\pm$ satisfy \cref{Equation3: One-sided DE} and \cref{Equation4: One-sided DE}, then
\[
\begin{pmatrix}
k_{\pm,1} \\
k_{\pm,2}
\end{pmatrix} := P_\pm(\rL)^{-1} \Phi_\pm(1)
= 
\frac{m_\pm}{\det P_\pm(\rL)} 
\begin{pmatrix}
\frac{1 \pm a(\rR)}{+1 \mp a(\rL)}\\
\frac{1 \mp a(\rR)}{-1 \mp a(\rL)}
\end{pmatrix}
\]
\end{lemma}
\begin{proof}
It follows from \cref{Equation: To Be Used Later} that
\[
\det P_\pm(\rL)
\begin{pmatrix}
k_{\pm,1} \\
k_{\pm,2}
\end{pmatrix} 
= 
m_\pm
\begin{pmatrix}
1 & -z_{\pm,2}(\rL) \\
-1 & z_{\pm,1}(\rL)
\end{pmatrix} 
\begin{pmatrix}
1\\
\frac{\pm \beta(0)}{\alpha_\mp(\rL)^*}
\end{pmatrix} =
m_\pm
\begin{pmatrix}
-\left(-1  +z_{\pm,2}(\rL)\frac{\pm \beta(0)}{\alpha_\mp(\rL)^*}\right) \\
-1 + z_{\pm,1}(\rL)\frac{\pm \beta(0)}{\alpha_\mp(\rL)^*}
\end{pmatrix} 
\]
With \cref{Equation: Little Trick} in mind, \cref{Equation: Eigenvalues of Apmx} becomes
\[
z_{\pm,j} = \frac{(1 \mp p)}{q} \frac{b(\rL)^*}{(-1)^j \mp a(\rL)},
\qquad j=1,2.
\]
We get
\[
z_{\pm,j}(\rL) \frac{\pm \beta(0)}{\alpha_\mp(\rL)^*} = \frac{(1 \mp p)}{q} \frac{b(\rL)^*}{(-1)^j \mp a(\rL)}\frac{|q|(\mp a(\rL) \mp a(\rR))}{(1 \mp p)e^{-i\theta} b(\rL)^*} 
= \frac{(-1)^{j+1} \mp a(\rR)}{(-1)^j \mp a(\rL)} + 1.
\]
Thus we obtain
\[
\begin{pmatrix}
k_{\pm,1} \\
k_{\pm,2}
\end{pmatrix} =
\frac{m_\pm}{\det P_\pm(\rL)} 
\begin{pmatrix}
\frac{1 \pm a(\rR)}{+1 \mp a(\rL)}\\
\frac{1 \mp a(\rR)}{-1 \mp a(\rL)}
\end{pmatrix}.
\]
\end{proof}

\begin{proof}[Proof of Equality \cref{Equation2: Classification of Type II}]
We shall first assume $a(\rR) = 1.$ If the sequences $\Phi_\pm$ satisfy \cref{Equation3: One-sided DE} and \cref{Equation4: One-sided DE}, then as before $\Phi_+$ is square summable (resp. $\Phi_-$ is square summable) if and only if
\[
m_+^2 \sum_{x=0}^\infty |z_{+,1}(\rL)|^{-2x} < \infty \qquad \left(\mbox{resp. } m_-^2 \sum_{x=0}^\infty |z_{-,2}(\rL)|^{-2x} < \infty \right),
\]
where \cref{Equation: Absolute Value of Zpm} together with \cref{Equation: Reflection Property} gives
\[
|z_{+,1}(\rL)|^{-1} = f(+ p)f(+ a(\rL)) \mbox{ and } |z_{-,2}(\rL)|^{-1} = f(- p)f(+ a(\rL)).
\]
Therefore, we obtain
\[
d_\pm = 
\begin{cases}
1, &\pm p + a(\rL) < 0, \\
0, &\pm p + a(\rL) \geq 0. \\
\end{cases}
\]
An analogous argument gives that if $a(\rR) = -1,$ then
\[
d_\pm = 
\begin{cases}
1, &\pm p - a(\rL) < 0, \\
0, &\pm p - a(\rL) \geq 0. \\
\end{cases}
\]
The claim follows.
\end{proof}

\subsubsection{Type \ref{Equation: Definition of Type III} coin} 
Let $C$ be of Type \ref{Equation: Definition of Type III}. This case turns out to be the hardest case. Here, we shall make use of the isomorphism \cref{Equation3: Isomorphism from Qepsilon}:
\[
d_\pm = \dim \ker Q_{\epsilon_\pm} = \dim \ker\left(L \oplus L - \bigoplus_{x \in \Z} A_\pm(x)\right).
\]

\begin{lemma}
\label{Lemma: Turning One DE into 2 DEs}
Given arbitrary $\C^2$-valued sequences $\Phi_\pm = (\Phi_\pm(x))_{x \in \Z},$ we define two sequences $\Phi_{\pm,\rL},\Phi_{\pm,\rR}$ by
\[
\Phi_{\pm,\rL}(x) := \Phi_\pm(-x+1) \mbox{ and } \Phi_{\pm,\rR}(x) := \Phi_\pm(x), \qquad x \in \N.
\]
Then $\Phi_\pm \in \ker\left(L \oplus L - \bigoplus_{x \in \Z} A_\pm(x)\right)$ if and only if the following three conditions are simultaneously satisfied:
\begin{align}
\label{Equation1: Criterion for Case C} &\Phi_{\pm,\rR} \in \ker\left(L \oplus L - \bigoplus_{x \in \N}^\infty A_\pm(\rR)\right),\\ 
\label{Equation2: Criterion for Case C} &\Phi_{\pm,\rL} \in \ker\left(L \oplus L - \bigoplus_{x \in \N}^\infty A_\pm(\rL)^{-1}\right),\\
\label{Equation3: Criterion for Case C} &\Phi_{\pm,\rR}(1)  = A_\pm(0)\Phi_{\pm,\rL}(1).
\end{align}
\end{lemma}
\begin{proof}
Evidently, we have $\Phi_\pm \in \ker\left(L \oplus L - \bigoplus_{x \in \Z} A_\pm(x)\right)$ if and only if the following three conditions are simultaneously satisfied:
\begin{align*}
&\Phi_\pm(+x+1) = A_\pm(\rR)\Phi_\pm(+x), \qquad  x \in \N, \\
&\Phi_\pm(-x+1) = A_\pm(\rL)\Phi_\pm(-x), \qquad  x \in \N, \\
&\Phi_\pm(0 + 1) = A_\pm(0)\Phi_\pm(0),
\end{align*}
where the last condition is obviously \cref{Equation3: Criterion for Case C} and the first two conditions are equivalent to the following two equations respectively:
\begin{align*}
&(L \oplus L) \Phi_{\pm,\rR}(x) = A_\pm(\rR)\Phi_{\pm,\rR}(x), \\
&(L \oplus L) \Phi_{\pm,\rL}(x) = A_\pm(\rL)^{-1}\Phi_{\pm,\rL}(x).
\end{align*}
The claim follows.
\end{proof}

\begin{lemma}
We have
\begin{equation}
\label{Equation: Sandwitch of A by P}
P_\pm(\rR)^{-1}A_\pm(0) P_\pm(\rL)  = 
\begin{pmatrix}
z_{\pm,1}(\rL) &  0\\
0  & z_{\pm,2}(\rL)
\end{pmatrix}.
\end{equation}
\end{lemma}
\begin{proof}
As in \cref{Equation: Little Trick}, we can let
\[
r_\pm = \frac{q^*}{1 \pm p} = \frac{1 \mp p}{q}.
\]
We have
\begin{align*}
P_\pm(\sharp) &= 
\begin{pmatrix}
z_{\pm,1}(\sharp) & z_{\pm,2}(\sharp) \\
1 & 1
\end{pmatrix} = 
\begin{pmatrix}
r_\pm \left(\frac{-1 \pm a(\sharp)}{b(\sharp)}\right) & r_\pm \left(\frac{+1 \pm a(\sharp)}{b(\sharp)}\right) \\
1 & 1
\end{pmatrix}, \\ 
A_\pm(0) &= 
\begin{pmatrix}
\frac{\mp \beta(0)}{\alpha_\pm(\rR)} &  \frac{\alpha_\mp(\rL)^*}{\alpha_\pm(\rR)} \\
1 & 0
\end{pmatrix}
=
\begin{pmatrix}
\frac{\pm |q|(a(\rR) + a(\rL))}{(1 \pm p) e^{i \theta} b(\rR)} &  \frac{(1 \mp p) e^{-i \theta} b(\rL)^*}{(1 \pm p) e^{i \theta} b(\rR)} \\
1 & 0
\end{pmatrix}
=
\begin{pmatrix}
\frac{\pm r_\pm (a(\rR) + a(\rL))}{b(\rR)} &  \frac{r_\pm^2 b(\rL)^*}{b(\rR)} \\
1 & 0
\end{pmatrix}.
\end{align*}
On one hand,
\begin{align*}
P_\pm(\rR) 
\begin{pmatrix}
z_{\pm,1}(\rL) &  0\\
0  & z_{\pm,2}(\rL)
\end{pmatrix}
&=
\begin{pmatrix}
z_{\pm,1}(\rR) & z_{\pm,2}(\rR) \\
1 & 1
\end{pmatrix}
\begin{pmatrix}
z_{\pm,1}(\rL) &  0\\
0  & z_{\pm,2}(\rL)
\end{pmatrix} \\ 
&=
\begin{pmatrix}
z_{\pm,1}(\rL)z_{\pm,1}(\rR) &  z_{\pm,2}(\rL)z_{\pm,2}(\rR)\\
z_{\pm,1}(\rL)  & z_{\pm,2}(\rL)
\end{pmatrix}.
\end{align*}
On the other hand,
\begin{align*}
A_\pm(0) P_\pm(\rL) &=
\begin{pmatrix}
\frac{\pm r_\pm (a(\rR) + a(\rL))}{b(\rR)} &  \frac{r_\pm^2 b(\rL)^*}{b(\rR)} \\
1 & 0
\end{pmatrix}
\begin{pmatrix}
r_\pm \left(\frac{-1 \pm a(\rL)}{b(\rL)}\right) & r_\pm \left(\frac{+1 \pm a(\rL)}{b(\rL)}\right) \\
1 & 1
\end{pmatrix} \\ 
&=
\begin{pmatrix}
\frac{\pm r_\pm^2 (-1 \pm a(\rL))(a(\rR) + a(\rL))}{b(\rL) b(\rR)} + \frac{r_\pm^2 |b(\rL)|^2}{b(\rL)b(\rR)} & \frac{\pm r_\pm^2 (+1 \pm a(\rL))(a(\rR) + a(\rL))}{b(\rL) b(\rR)} + \frac{r_\pm^2 |b(\rL)|^2}{b(\rL)b(\rR)} \\
z_{\pm,1}(\rL) & z_{\pm,2}(\rL)
\end{pmatrix}
,
\end{align*}
where for each $j = 1,2,$ we have
\begin{align*}
&\frac{\pm r_\pm^2 ((-1)^j \pm a(\rL))(a(\rR) + a(\rL))}{b(\rL) b(\rR)} + \frac{r_\pm^2 |b(\rL)|^2}{b(\rL)b(\rR)} \\
&= \frac{r_\pm^2}{b(\rL) b(\rR)} \left( \pm (-1)^j a(\rR)  \pm (-1)^ja(\rL) + a(\rL)a(\rR) + a(\rL)^2 + |b(\rL)|^2\right)  \\
&= \frac{r_\pm^2}{b(\rL) b(\rR)} \left( \pm (-1)^j a(\rR)  \pm (-1)^ja(\rL) + a(\rL)a(\rR) + 1\right)  \\
&= \frac{r_\pm^2}{b(\rL) b(\rR)} ((-1)^j \pm a(\rL))((-1)^j \pm a(\rR))\\
&= z_{\pm,j}(\rL)z_{\pm,j}(\rR).
\end{align*}
We obtain \cref{Equation: Sandwitch of A by P} as follows;
\[
A_\pm(0) P_\pm(\rL) = 
\begin{pmatrix}
z_{\pm,1}(\rL)z_{\pm,1}(\rR) & z_{\pm,2}(\rL)z_{\pm,2}(\rR) \\
z_{\pm,1}(\rL) & z_{\pm,2}(\rL)
\end{pmatrix} = 
P_\pm(\rR) 
\begin{pmatrix}
z_{\pm,1}(\rL) &  0\\
0  & z_{\pm,2}(\rL)
\end{pmatrix}.
\]
\end{proof}

\begin{corollary}
\label{Corollary: Square Summability}
Suppose that $\C^2$-valued sequences $\Phi_\pm$ satisfy algebraic equations $\Phi_\pm(x+1) = A_\pm(x)\Phi_\pm(x)$ for each $x \in \Z,$ and that
\[
\begin{pmatrix}
k_{\pm,1} \\
k_{\pm,2} \\
\end{pmatrix} := 
P_\pm(\rL)^{-1} \Phi_\pm(0).
\]
Then $\Phi_\pm \in \ker\left(L \oplus L - \bigoplus_{x \in \Z} A_\pm(x)\right)$ if and only if the following sum is finite for each $j=1,2:$
\begin{equation}
\label{Equation: Type III Criterion}
|k_{\pm,j}|^2 \sum_{x \in \N} \left( |z_{\pm,j}(\rL)|^{-2x} +  |z_{\pm,j}(\rR)|^{+2x}\right) < \infty.
\end{equation}
Moreover, 
\begin{align}
\label{Equation1: Critical Case} |z_{+,1}(\rL)|^{-1} < 1 \mbox{ and } |z_{+,1}(\rR)| < 1 \mbox{ if and only if } -p \in (a(\rL),a(\rR)), \\
\label{Equation2: Critical Case} |z_{+,2}(\rL)|^{-1} < 1 \mbox{ and } |z_{+,2}(\rR)| < 1 \mbox{ if and only if } +p \in (a(\rR),a(\rL)), \\
\label{Equation3: Critical Case} |z_{-,1}(\rL)|^{-1} < 1 \mbox{ and } |z_{-,1}(\rR)| < 1 \mbox{ if and only if } -p \in (a(\rR),a(\rL)), \\
\label{Equation4: Critical Case} |z_{-,2}(\rL)|^{-1} < 1 \mbox{ and } |z_{-,2}(\rR)| < 1 \mbox{ if and only if } +p \in (a(\rL),a(\rR)). 
\end{align}
Furthermore, we have $d_\pm \leq 1.$
\end{corollary}
Note that $k_{\pm,1}$ and $k_{\pm,2}$ cannot be simultaneously both non-zero. 
\begin{proof}
Recall that $\Phi \in \ker\left(L \oplus L - \bigoplus_{x \in \Z} A_\pm(x)\right)$ if and only if \cref{Equation1: Criterion for Case C,Equation2: Criterion for Case C,Equation3: Criterion for Case C} hold true. With the notation introduced in \cref{Lemma: Turning One DE into 2 DEs} in mind, we have
\begin{align*}
&
\begin{pmatrix}
k_{\pm,1} \\
k_{\pm,2} \\
\end{pmatrix} =
P_\pm(\rL)^{-1}\Phi_\pm(0) = P_\pm(\rL)^{-1}\Phi_{\pm,\rL}(1)
, \\
&
\begin{pmatrix}
k'_{\pm,1} \\
k'_{\pm,2} \\
\end{pmatrix}
:=
P_\pm(\rR)^{-1} \Phi_\rR(1) 
=
P_\pm(\rR)^{-1}A_\pm(0)\Phi_{\pm,\rL}(1)
= 
\begin{pmatrix}
z_{\pm,1}(\rL)k_{\pm,1}\\
z_{\pm,2}(\rL)k_{\pm,2}
\end{pmatrix},
\end{align*}
where the second last equality follows from \cref{Equation3: Criterion for Case C} and the last equality follows from \cref{Equation: Sandwitch of A by P}. It follows from \cref{Lemma: Square Summability} that $\Phi \in \ker\left(L \oplus L - \bigoplus_{x \in \Z} A_\pm(x)\right)$ if and only if
\begin{align*}
& |k_{\pm,j}|^2  \sum_{x \in \N} |z_{\pm,j}(\rL)|^{-2x} < \infty, &&j=1,2, \\
&|k_{\pm,j}|^2 |z_{\pm,j}(\rL)|^2 \sum_{x \in \N}  |z_{\pm,j}(\rR)|^{+2x} < \infty, &&j=1,2,
\end{align*}
and so we get the criterion \cref{Equation: Type III Criterion}. It follows from \cref{Equation: Absolute Value of Zpm} that
\[
\begin{matrix}
|z_{\pm,1}(\rL)|^{-1} = f(\pm p)f(\pm a(\rL)), \qquad & |z_{\pm,1}(\rR)| = f(\mp p)f(\mp a(\rR)), \\
|z_{\pm,2}(\rL)|^{-1} = f(\pm p)f(\mp a(\rL)), \qquad & |z_{\pm,2}(\rR)| = f(\mp p)f(\pm a(\rR)),
\end{matrix}
\]
where 
\begin{align*}
&|z_{\pm,1}(\rL)|^{-1} < 1 \mbox{ and } |z_{\pm,1}(\rR)| < 1 \mbox{ if and only if }  \pm a(\rL) < \mp p < \pm a(\rR), \\
&|z_{\pm,2}(\rL)|^{-1} < 1 \mbox{ and } |z_{\pm,2}(\rR)| < 1 \mbox{ if and only if }  \mp a(\rL) < \mp p < \mp a(\rR).
\end{align*}
Thus \crefrange{Equation1: Critical Case}{Equation4: Critical Case} hold true. Finally, we may assume without loss of generality that $a(\rL) < a(\rR),$ so that \crefrange{Equation2: Critical Case}{Equation3: Critical Case} both fail to hold. That is, $\Phi_\pm \in \ker\left(L \oplus L - \bigoplus_{x \in \Z} A_\pm(x)\right)$ are always of the following forms:
\[
\begin{pmatrix}
k_{+,1} \\
0 \\
\end{pmatrix} = 
P_\pm(\rL)^{-1} \Phi_+(0) \mbox{ and }
\begin{pmatrix}
0 \\
k_{-,2} \\
\end{pmatrix} = 
P_\pm(\rL)^{-1} \Phi_-(0),
\]
and so $d_\pm \neq 2.$ This is because the conclusion of \cref{Lemma: Linear Independence of Solutions to Matrix Difference Equation} still holds true, if the indexing set $\N$ is replaced by $\Z.$
\end{proof}

\begin{proof}[Proof of Equality \cref{Equation: Classification of Type III}]
The claim immediately follows from \crefrange{Equation1: Critical Case}{Equation4: Critical Case}.
\end{proof}

\section{Proof of the main theorem} 
\label{Section: Proof of the main theorem}

We are finally in a position to prove \cref{Theorem: MainTheorem}, the main theorem of the present paper. This will be done in two separate steps: proof of the Fredholmness characterisation as in \cref{Equation: Fredholmness} and proof of the index formula \cref{Equation: Index Formula for Type III}.

\subsection{The Fredholmness}
\label{Section: Characterisation of Closed Range}

In order to prove the Fredholmness characterisation \cref{Equation: Fredholmness}, let us first discuss the following simple characterisation of the closedness of the range of $Q_{\epsilon_+}$;
\begin{lemma}
\label{Remark: Spectral Gaps and Closedness}
With the notation introduced in \cref{Lemma: Diagonalisation of Shift Implies Off-diagonalisation of Super Charge} in mind, the operator $Q_{\epsilon_+}$ has a closed range if and only if the time-evolution $U$ has spectral gaps\footnote 
{
That is to say, $\pm 1$ are not accumulation points of the spectrum of $U.$
} 
at $\pm 1.$
\end{lemma}
\begin{proof}
It is a well-known fact that $Q_{\epsilon_+}$ has a closed range if and only if $\inf \sigma(Q_{\epsilon_+}^*Q_{\epsilon_+}) \setminus \{0\} > 0$ (see, for example, \cite[Lemma 7.27]{Book:Arai17:AnalysisOnFockSpacesAndMathematicalTheoryOfQuantumFields}). Since $\sigma(H_\epsilon) = \sigma(H)$ and $H = ( \textrm{Im\,} U)^2,$ the claim follows from the spectral theorem.
\end{proof}

To put it another way, \cref{Remark: Spectral Gaps and Closedness} states that the closedness of the operator $Q_{\epsilon_+}$ is nothing but a spectral property of the time-evolution $U = \varGamma C.$ The following result is therefore useful;

\begin{theorem}
\label{Theorem: Essential Spectrum of Time-evolution}
Let $(\varGamma,C)$ be a split-step SUSYQW with an anisotropic coin $C,$ and let
\[
U_\sharp = \varGamma \bigoplus_{x \in \Z} C(\sharp),
\qquad \sharp = \rL,\rR.
\]
Then the essential spectrum of the time-evolution $U = \varGamma C$ is given by
\begin{equation}
\label{Equation: Essential Spectrum of U}
\ess(U) = \sigma(U_\rL) \cup \sigma(U_\rR).
\end{equation}
More explicitly, we have $\sigma(U_\sharp) = \{z \in \T \mid \Re z \in I_\sharp\}$ for each $\sharp = \rL,\rR,$ where
\begin{equation}
\label{Equation: Spectrum of Usharp}
I_\sharp := 
\begin{cases}
\{\pm 1\}, & \mbox{if $C(\sharp)$ is trivial}, \\
[pa(\sharp) - |qb(\sharp)|,pa(\sharp) + |qb(\sharp)|], & \mbox{otherwise}. 
\end{cases}
\end{equation}
\end{theorem}
\begin{proof}[Proof of \cref{Theorem: Essential Spectrum of Time-evolution}]
This result is standard, and so we will only give a brief sketch of the proof. Firstly, the well-known equality \cref{Equation: Essential Spectrum of U} can be proved by either using Weyl's criterion for the essential spectrum or an elegant 
$C^*$-algebraic approach (see, for example, \cite[Theorem 2.2]{Richard-Suzuki-Aldecoa18}).  Secondly, the fact that each $\sigma(U_\sharp)$ is characterised by \cref{Equation: Spectrum of Usharp} is an easy consequence of the standard approach which makes use of the discrete Fourier transform.
\end{proof}

We are now in a position to show that \cref{Equation: Fredholmness} is also a characterisation of the closedness of the range of $Q_{\epsilon_+}$;
\begin{theorem}
\label{Theorem: Closedness of Range of Qepsilon}
Let $(\varGamma,C)$ be a split-step SUSYQW endowed with an isotropic coin $C.$ With the notation introduced in \cref{Theorem: Witten Index Formula} in mind, the operator $Q_{\epsilon_+}$ has a closed range if and only if $|p| \neq |a(\sharp)|$ whenever $C(\sharp)$ is non-diagonal, where $\sharp = \rL,\rR.$
\end{theorem}
\begin{proof}
It immediately follows from \cref{Theorem: Essential Spectrum of Time-evolution} that for each $\sharp = \rL,\rR,$ the set $\sigma(U_\sharp)$ is a discrete subset of $\T$ if and only if $C(\sharp)$ is a diagonal matrix (i.e. $b(\sharp) = 0$). With  \cref{Remark: Spectral Gaps and Closedness} in mind, we have that the time-evolution $U$ has spectral gaps at $\pm 1$ if and only if the set $I_\sharp$ does not contain both $-1$ and $+1,$ whenever the limit $C(\sharp)$ is not a diagonal matrix. For such $\sharp,$ we introduce the following parametrisation:
\[
\Theta := \arcsin(p) \mbox{ and } \Theta_\sharp := \arcsin(a(\sharp)).
\]
With this parametrisation in mind, we obtain
\begin{align*}
p &= \sin \Theta, & q &= e^{i \Arg q}\cos \Theta , \\
a(\sharp) &= \sin \Theta_\sharp, & b(\sharp) &=  e^{i \Arg b(\sharp)}\cos \Theta_\sharp.
\end{align*}
The addition formula for the cosine gives:
\[
pa(\sharp) \pm |qb(\sharp)| = \sin \Theta \sin \Theta_\sharp \pm  \cos \Theta \cos \Theta_\sharp = \pm \cos (\Theta \mp \Theta_\sharp),
\]
so that the set $I_\sharp$ becomes the following closed interval:
\begin{equation}
I_\sharp = [-\cos (\Theta + \Theta_\sharp), +\cos (\Theta - \Theta_\sharp)],
\end{equation}
where $-\pi < \Theta \pm \Theta_\sharp < \pi.$ Thus, we have $\pm 1 \in I_\sharp$ if and only if $\Theta \mp \Theta_\sharp = 0.$ That is, $I_\sharp$ does not contain both $-1$ and $+1$ if and only if $|p| \neq |a(\sharp)|.$ The claim follows.
\end{proof}

\begin{proof}[Proof of the characterisation \cref{Equation: Fredholmness}]
Let $(\varGamma,C)$ be a split-step SUSYQW endowed with an anisotropic coin $C,$ and let $C(\sharp)$ be a non-trivial unitary involution for each $\sharp = \rL,\rR.$ Recall the Fredholmness is invariant under compact perturbations (see \cref{Equation: Topological Invariance} for details). Therefore, we may assume without loss of generality that $C$ is of the form \cref{Equation: Simplified Coin}, so that \cref{Theorem: Classification of Dimensions} implies $d_\pm \leq 1.$ That is, \cref{Theorem: Closedness of Range of Qepsilon} implies that $(\varGamma,C)$ is Fredholm if and only if $|p| \neq |a(\sharp)|$ for each $\sharp = \rL,\rR.$ Note that if $C(\sharp)$ is diagonal, then $|p| \neq |a(\sharp)| = 1$ obviously holds true. The claim follows.
\end{proof}

\subsection{Proof of the index formula \cref{Equation: Index Formula for Type III}}
From here on, we shall assume that $(\varGamma,C)$ is a Fredholm SUSYQW and prove \cref{Equation: Index Formula for Type III} by considering the four coin types separately:

\subsubsection{Type \ref{Equation: Definition of Type I} coin} 

\begin{proof}[Proof of Equality \cref{Equation: Index Formula for Type III}]
If $C$ is of Type \ref{Equation: Definition of Type I}, then \cref{Equation: Index Formula for Type III} becomes
\begin{equation}
\label{Equation: Index Formula for Type I} 
\ind(\varGamma,C) = 0,
\end{equation}
since $|a(\sharp)| = 1$ for each $\sharp = \rL,\rR$ by assumption. In fact, \cref{Equation: Index Formula for Type I} immediately follows \cref{Equation: Classification of Type I}. The claim follows.
\end{proof}

\subsubsection{Type \ref{Equation1: Definition of Type II} coin} 

\begin{proof}[Proof of Equality \cref{Equation: Index Formula for Type III}]
If $C$ is of Type \ref{Equation1: Definition of Type II}, then \cref{Equation: Index Formula for Type III} becomes
\begin{equation}
\label{Equation1: Index Formula for Type II} 
\ind(\varGamma,C) = 
\begin{cases}
+\sgn p, & |a(\rR)| < |p|, \\
0, & \mbox{otherwise}, 
\end{cases} 
\end{equation}
since $|a(\rL)| = 1$ by assumption. It follows from \cref{Theorem: Classification of Dimensions} that $d_+ \neq d_-$ if and only if one of the following two conditions holds true:
\[
(d_\pm,d_\mp) = (1,0) \mbox{ if and only if } \mp p \leq a(\rL)a(\rR) < \pm p,
\]
where the last condition is equivalent to $|a(\rR)| < \pm p$, since $|p| \neq |a(\rR)|.$ Thus \cref{Equation1: Index Formula for Type II} holds true.
\end{proof}

\subsubsection{Type \ref{Equation2: Definition of Type II} coin} 

\begin{proof}[Proof of Equality \cref{Equation: Index Formula for Type III}]
If $C$ is of Type \ref{Equation2: Definition of Type II}, then \cref{Equation: Index Formula for Type III} becomes
\begin{equation}
\label{Equation2: Index Formula for Type II} 
\ind(\varGamma,C) = 
\begin{cases}
-\sgn p, & |a(\rL)| < |p|, \\
0, & \mbox{otherwise}, 
\end{cases} 
\end{equation}
since $|a(\rR)| = 1$ by assumption. It follows from \cref{Theorem: Classification of Dimensions} that $d_+ \neq d_-$ if and only if one of the following two conditions holds true:
\[
(d_\pm,d_\mp) = (1,0) \mbox{ if and only if } \pm p \leq a(\rL)a(\rR) < \mp p,
\]
where the last condition is equivalent to $|a(\rR)| < \mp p$, since $|p| \neq |a(\rR)|.$
\end{proof}

\subsubsection{Type \ref{Equation: Definition of Type III} coin} 

\begin{proof}[Proof of Equality \cref{Equation: Index Formula for Type III}]
Let us assume that $C$ is of Type \ref{Equation: Definition of Type III} coin operator. Note first that if $|a(\rL)| = |a(\rR)|,$ then $\ind(\varGamma,C) = 0.$ Thus, we shall assume $|a(\rL)| \neq |a(\rR)|$ from here on. By the invariance principle \cref{Equation: Witten Index for Negative Coin}, we shall assume without loss of generality that $a(\rR) < a(\rL)$ throughout, and so
\[
\ind(\varGamma,C) = \pm 1 \mbox{ if and only if } \pm p \in (a(\rR),a(\rL)) \mbox{ and } \mp p \notin (a(\rR),a(\rL)).
\]
We show first that \cref{Equation: Index Formula for Type III} holds true, under the assumption $|a(\rR)| < |a(\rL)|$ first. That is, we need to check 
\[
\ind(\varGamma,C) = 
\begin{cases}
+ \sgn p, & \mbox{if } p \neq 0 \mbox{ and } |a(\rR)| < |p| <|a(\rL)|, \\
0, & \mbox{otherwise}
\end{cases}
\]
Since $|a(\rR)| < |a(\rL)|$ and $a(\rR) < a(\rL),$ we must always have $a(\rL) = |a(\rL)|.$ If $a(\rR) \geq 0,$ then 
\[
\ind(\varGamma,C) = \pm 1 \mbox{ if and only if }  a(\rR) < \pm p < a(\rL).
\]
On the other hand, if $a(\rR) < 0,$ then
\[
\ind(\varGamma,C) = \pm 1 \mbox{ if and only if } \pm p \in (|a(\rR)|,|a(\rL)|).
\]
It remains to prove that \cref{Equation: Index Formula for Type III} holds true, under the other assumption $|a(\rL)| < |a(\rR)|.$ That is, we need to check 
\[
\ind(\varGamma,C) = 
\begin{cases}
- \sgn p, & \mbox{if } p \neq 0 \mbox{ and } |a(\rL)| < |p| < |a(\rR)|, \\
0, & \mbox{otherwise}. 
\end{cases}
\]
As before invariance principle \cref{Equation: Witten Index for Negative Coin} allows us to assume without loss of generality that $a(\rR) < a(\rL).$ Since $|a(\rL)| < |a(\rR)|$ and $a(\rR) < a(\rL),$ we must always have $a(\rR) = -|a(\rR)|.$ If $a(\rL) < 0,$ then 
\[
\ind(\varGamma,C) = \pm 1 \mbox{ if and only if } |a(\rL)| < \mp p < |a(\rR)|.
\]
On the other hand, if $a(\rL) \geq 0,$ then
\[
\ind(\varGamma,C) = \pm 1 \mbox{ if and only if }  |a(\rL)| < \mp p < |a(\rR)|.
\]
\end{proof}

\section{Concluding Remarks}
\label{Section: Concluding Remarks}

A somewhat natural question arises. Can we still define the Witten index, if a given SUSYQW fails to be Fredholm? The answer to this question turns out to be yes, and we shall give a brief account of how research towards this direction can be undertaken. In fact, the standard theory of supersymmetry is already capable of dealing with the Witten index which cannot be interpreted as the Fredholm index by making use of a certain trace formula. See, for example, \cite{BGGSS87} or \cite{Gesztesy-Simon88}. These papers provide a theoretical foundation in what follows.

Let $(\varGamma,C)$ be a one-dimensional split-step SUSYQW, and let $\epsilon$ be any unitary operator which gives diagonalisation of the shift operator as in \cref{Equation: Diagonalisation of the Shift Operator}. We can then consider the unitarily equivalent SUSYQW $(\epsilon^* \varGamma \epsilon,\epsilon^*C \epsilon)$ together with the new associated supercharge $Q_\epsilon := \epsilon^* \varGamma \epsilon$ and superhamiltonian $H_\epsilon := \epsilon^* H \epsilon = Q_\epsilon^2$ admitting:
\[
Q_\epsilon = 
\begin{pmatrix}
0 & Q_{\epsilon_-} \\
Q_{\epsilon_+} & 0
\end{pmatrix}, \qquad
H_\epsilon =
\begin{pmatrix}
H_{\epsilon_+} & 0 \\
0 & H_{\epsilon_-}
\end{pmatrix},
\]
where the first equality follows from \cref{Equation: Skew-diagonal Representation of Qepsilon}. We say that the triple $(\varGamma,C,\epsilon)$ is \textbi{trace-compatible}, if $H_{\epsilon_+} - H_{\epsilon_-}$ is a trace-class operator on $\ell^2(\Z).$ We can then define the \textbi{Witten index} of the triple $(\varGamma,C,\epsilon)$ by
\begin{equation}
\label{Equation: Trace-formula for The Witten Index}
\ind(\varGamma,C,\epsilon) := \lim_{t \to \infty} \tr(e^{-t H_{\epsilon_+}} - e^{-t H_{\epsilon_-}}),
\end{equation}
whenever the limit exists. It is not known to the authors whether or not Formula \cref{Equation: Trace-formula for The Witten Index} depends on $\epsilon.$ As in \cite{BGGSS87}, if the SUSYQW $(\varGamma,C)$ turns out to be Fredholm, then the above limit exists, and we get
\[
\ind(\varGamma,C,\epsilon) = \ind(\varGamma,C),
\]
where the left hand side does not depend on $\epsilon$ in this Fredholm case. That is, in principle, we should be able to recover \cref{Equation: Index Formula for Type III} by simply evaluating the trace-formula \cref{Equation: Trace-formula for The Witten Index}. Research towards this direction is work in progress, and this will be part of the PhD dissertation of the second author. The present paper concludes with the following simple example;

{\footnotesize
\begin{example}
Let $(\varGamma,C)$ be a one-dimensional split-step SUSYQW whose coin operator $C$ has the property that $b(x) = 0$ for each $x \in \Z.$ With the notation introduced in \cref{Theorem: Witten Index Formula} in mind, we obtain
\[
-2iQ_{\epsilon_\pm} = 0 - 0 \pm |q|(a_2(\cdot + 1) - a_1) =: \pm \beta.
\]
Then the superhamiltonian becomes
\[
H_\epsilon = Q_{\epsilon}^2 = 
\begin{pmatrix}
0 & -\frac{i\beta}{2} \\
+\frac{i\beta}{2}& 0
\end{pmatrix}
\begin{pmatrix}
0 & -\frac{i\beta}{2} \\
+\frac{i\beta}{2}& 0
\end{pmatrix}
=
\begin{pmatrix}
\frac{\beta^2}{4} & 0\\
0& \frac{\beta^2}{4} 
\end{pmatrix}.
\]
This implies $H_+ = H_-,$ so that $(\varGamma,C,\epsilon)$ is trace-compatible and $\ind(\varGamma,C,\epsilon) = 0.$
\end{example}
}

\appendix
\def\thesection{\Roman{section}}
\section{Supplementary Material}
\label{Section: Appendix}

An (abstract) \textbi{supersymmetric quantum walk} (SUSYQW) is a pair $(\varGamma,C)$ of two unitary involutions on a Hilbert space $\rH.$


The Witten index for SUSYQWs turns out to enjoy the following two invariance principles, each of which will a significant role in this paper.

\begin{theorem}[Invariance of the Witten index]
\label{Theorem: Invariance of the Witten Index}
The following two assertions hold true:
\begin{enumerate}
\item Unitary Invariance. Let $(\varGamma,C)$ and $(\varGamma',C')$ be two SUSYQWs that are \textbi{unitarily equivalent} in the sense 
$
(\varGamma',C') = (\epsilon^* \varGamma \epsilon, \epsilon^* C \epsilon)
$
for some  unitary operator $\epsilon$ on $\rH.$ 

Then $(\varGamma,C)$ is a Fredholm SUSYQW if and only if so is $(\epsilon^* \varGamma \epsilon, \epsilon^* C \epsilon).$
In this case,  
\begin{equation}
\label{Equation: Unitary Invariance}
\ind(\varGamma,C) = \ind(\epsilon^* \varGamma \epsilon, \epsilon^* C \epsilon).
\end{equation}
\item Topological Invariance. Let $(\varGamma,C)$ and $(\varGamma,C')$ be two SUSYQWs sharing the same shift operator $\varGamma,$ and let $C - C'$ be a compact operator. Then $(\varGamma,C)$ is a Fredholm SUSYQW if and only if so is $(\varGamma, C').$ In this case, 
\begin{equation}
\label{Equation: Topological Invariance}
\ind(\varGamma,C) = \ind(\varGamma,C').
\end{equation}
\end{enumerate}
\end{theorem}

\begin{remark}
The invariance principle \cref{Equation: Unitary Invariance} can be used to classify SUSYQWs in the following precise sense. If the Witten indices associated with two given SUSYQWs $(\varGamma,C)$ and $(\varGamma',C')$ do not agree to each other, then they cannot be unitarily equivalent. This is, of course, analogous to the manner in which we use the homotopy/homology groups to prove that certain topological spaces are not homotopy equivalent.
\end{remark}

Here is yet another important principle of the Witten index:
\begin{theorem}[{\cite[Corollary 3.7]{Suzuki18}}]
If one of $(\varGamma,C),(-\varGamma,C),(\varGamma,-C)$ is a Fredholm SUSYQW, then so are the rest. In this case we have the following formulas:
\begin{align}
\label{Equation: Witten Index for Negative Coin}
\ind(\varGamma,-C) &= \ind(\varGamma,C), \\
\label{Equation: Witten Index for Negative Shift}
\ind(-\varGamma,C) &= - \ind(\varGamma,C).
\end{align}
\end{theorem}

\begin{acknowledgements}
This work was supported by the Research Institute for Mathematical Sciences, a Joint Usage/Research Center located in Kyoto University.
A. S. was supported by JSPS
KAKENHI Grant Number JP18K03327. 
\end{acknowledgements}


\bibliographystyle{alpha}
\newcommand{\etalchar}[1]{$^{#1}$}

\end{document}